\documentclass[11pt]{article}
\usepackage{multirow,microtype, algorithm, algpseudocode}%if unwanted, comment out or use option "draft"

\usepackage[tmargin=1in,bmargin=1in,lmargin=1in,rmargin=1in]{geometry}

\usepackage{enumitem}
\setlist[enumerate]{itemsep=0mm}
\setlist[itemize]{itemsep=0mm}

\usepackage{titlesec}

\titlespacing*{\section}
  {0pt}{0.6\baselineskip}{0.4\baselineskip}
\titlespacing*{\subsection}
  {0pt}{0.6\baselineskip}{0.4\baselineskip}

\title{Frameworks for Designing In-place Graph Algorithms}

\author{Sankardeep Chakraborty$^1$, Anish Mukherjee$^2$, Venkatesh Raman$^3$, Srinivasa Rao Satti$^4$ \\
\\
$^1$ National Institute of Informatics, Tokyo, Japan, sankardeep@nii.ac.jp \\
$^2$ Chennai Mathematical Institute, Chennai, India, anish@cmi.ac.in\\
$^3$ The Institute of Mathematical Sciences, HBNI, Chennai, India, vraman@imsc.res.in \\
$^4$ Seoul National University, Seoul, South Korea, ssrao@cse.snu.ac.kr
}
\date{}

\usepackage{graphicx}
\usepackage{amsfonts}
\usepackage{amssymb}
\usepackage{amsmath}

\usepackage{amsthm}
\theoremstyle{plain}
\newtheorem{theorem}{Theorem}

\newtheorem{claim}{Claim}
\newtheorem{lemma}{Lemma}
\newtheorem*{lemma*}{Lemma}
\newtheorem*{theorem*}{Theorem}

\newcommand{\Log}{\mbox{{\sf L}}}

\newcommand{\NL}{\mbox{{\sf NL}}}
\newcommand{\UL}{\mbox{{\sf UL}}}
\newcommand{\coUL}{\mbox{{\sf co-UL}}}

\newcommand{\NP}{\mbox{{\sf NP}}}

\newcommand{\rotated}{\mbox{{\sf rotate}}}
\newcommand{\implicit}{\mbox{{\sf implicit}}}

\newcommand{\white}{\mbox{{\sf white}}}
\newcommand{\gray}{\mbox{{\sf gray}}}
\newcommand{\black}{\mbox{{\sf black}}}
\newcommand{\grayone}{\mbox{{\sf gray1}}}
\newcommand{\graytwo}{\mbox{{\sf gray2}}}

\begin{document}

\maketitle

Read-only memory (ROM) model is a classical model of computation to study time-space tradeoffs of algorithms. One of the classical results on the ROM model is that any sorting algorithm that uses $O(s)$ words of extra space requires $\Omega (n^2/s)$ comparisons for
$ \lg n \leq s \leq n/\lg n$\footnote{We use $\lg$ to denote logarithm to the base $2$} and the bound has also been recently matched by an algorithm. However, if we relax the model (from ROM), we do have sorting algorithms (say Heapsort) that can sort using $O(n \lg n)$ comparisons using $O(\lg n)$ bits of extra space, even keeping a permutation of the given input sequence at any point of time during the algorithm.

We address similar questions for graph algorithms. We show that a simple natural relaxation of ROM model allows us to implement fundamental graph search methods like BFS and DFS more space efficiently than in ROM. By simply allowing elements in the adjacency list of a vertex to be permuted, we show that, on an undirected or directed connected graph $G$ having $n$ vertices and $m$ edges, the vertices of $G$ can be output in a DFS or BFS order using $O(\lg n)$ bits of extra space and $O(n^3 \lg n)$ time.
%\begin{itemize}
%\item
%DFS order using $O(\lg n)$ bits of extra space and $O(m^2/n)$ time if the graph is given in an adjacency list, and in $O(m^2 \lg n/n)$ time if the graph is given in an adjacency array;
%\item
%BFS order using $O(\lg n)$ bits of extra space and $O(m)$ time if all vertices have degree at least $2 \lg n +3$, in $O(n^2)$ time if there are no degree $2$ vertices, and in $O(n^3)$ time otherwise.
%\end{itemize}
%Most of these results carry over to directed graphs too, with a slight degradation in running time.
Thus we obtain similar bounds for {\it reachability} and {\it shortest path distance} (both for undirected and directed graphs). 
With a little more (but still polynomial) time, we can also output vertices in the {\it lex-DFS} order.
As reachability in directed graphs (even in DAGs) and shortest path distance (even in undirected graphs) are {\sf NL}-complete, and lex-DFS is {\sf P}-complete, our results show that our model is more powerful than ROM if {\sf L} $\neq$ {\sf P}.

En route, we also introduce and develop algorithms for another relaxation of ROM where the adjacency lists of the vertices are circular lists and we can modify only the heads of the lists. Here we first show a linear time DFS implementation using $n + O(\lg n)$ bits of extra space. Improving the extra space further to only $O(\lg n)$ bits, we also obtain BFS and DFS albeit with a slightly slower running time. Some of these algorithms also translate to improved algorithms for DFS and its applications in ROM. Both the models we propose maintain the graph structure throughout the algorithm, only the order of vertices in the adjacency list changes. In sharp contrast, for BFS and DFS, to the best of our knowledge, there are no algorithms in ROM that use even $O(n^{1-\epsilon})$ bits of extra space; in fact, implementing DFS using $cn$ bits for $c<1$ has been mentioned as an open problem.  Furthermore, DFS (BFS) algorithms using $n+o(n)$ ($o(n)$) bits of extra use Reingold's [JACM, 2008] or Barnes et al's reachability algorithm [SICOMP, 1998] and hence have high runtime. Our results can be contrasted with the recent result of Buhrman et al.~[STOC, 2014] which gives an algorithm for {\it directed st-reachability} on {\it catalytic Turing machines} using $O(\lg n)$ bits with catalytic space $O(n^2 \lg n)$ and time $O(n^9)$.

All our algorithms are simple but quite subtle, and we believe that these models are practical enough to spur interest for other graph problems in these models.

\tableofcontents

\newpage
\section{Introduction}
Motivated by the rapid growth of huge data set (``big data''), space efficient algorithms are becoming increasingly important than ever before. The proliferation of specialized handheld devices and embedded systems that have a limited supply of memory provide another motivation to consider space efficient algorithms. 
%As a consequence, algorithms that are oblivious to space constraint are not desired in such scenario. 
%
To design {\it space-efficient} algorithms in general, several models of computation have been proposed. Among them, the following two computational models have received considerable attention in the literature.
\begin{itemize}
\item In the \emph{read-only} memory (ROM) model, we assume that the input is given in a read-only memory. 
%(and so cannot be modified). 
The output of an algorithm is written on to a separate write-only memory, and the output can not be read or modified again. 
% So the input is ``read-only'' and the output is ``write-only''. 
In addition to the input and output media, a limited random access workspace is available. 
%We count space in terms of the number of bits of workspace used by an algorithm. 
Early work on this model was on designing lower bounds~\cite{Beame91,BorodinC82,BorodinFKLT81}, for designing algorithms for selection and sorting~\cite{ChanMR13,ElmasryJKS14,Frederickson87,MunroP80,MunroR96,PagterR98} and problems in computational geometry~\cite{AsanoBBKMRS14,AsanoMRW11,BarbaKLSS15,ChanC07,DarwishE14}. Recently there has been interest on space-efficient graph algorithms \cite{AsanoIKKOOSTU14,BCR,BanerjeeCRRS2015,ChakrabortyRS16,CRS17,ElmasryHK15,HagerupK16,KammerKL16}.

\item In the~\emph{in-place} model, the input elements are given in an array, and the algorithm may use the input array as working space. Hence the algorithm may modify the array during its execution. After the execution, all the input elements should be present in the array (maybe in a permuted order) and the output maybe put in the same array or sent to an output stream. The amount of extra space usage during the entire execution of the algorithm is limited to $O(\lg n)$ bits.
%$O(1)$ words.
%, although sometimes poly-logarithmic words of extra space is also allowed. 
A prominent example of an in-place algorithm is the classic heap-sort. %, which requires just $O(\lg n)$ bits
%$O(1)$ words of extra space. 
Other than in-place sorting \cite{FranceschiniMP07}, searching~\cite{FranceschiniM06,Munro86} and selection \cite{LaiW88} algorithms, many in-place algorithms have been designed in areas such as computational geometry~\cite{BronnimannCC04} and stringology~\cite{FranceschiniM07}. 
\end{itemize}

Apart from these models, researchers have also considered (semi)-streaming models~\cite{AlonMS99,FeigenbaumKMSZ05,MunroP80} for designing space-efficient algorithms. Very recently the following two new models were introduced in the literature with the same objective. 
%In what follows we briefly describe them.
\begin{itemize}
\item Chan et al.~\cite{ChanMR14} introduced the~\emph{restore} model which is a more relaxed version of read-only memory (and a restricted version of the in-place model), where the input is allowed to be modified, but at the end of the computation, the input has to be restored to its original form. They also gave space efficient algorithms for selection and sorting on integer arrays in this model. This has motivation, for example, in scenarios where the input (in its original form) is required by some other application.

\item Buhrman et al.~\cite{BuhrmanCKLS14,BuhrmanKLS16,Koucky16} introduced and studied the {\it catalytic-space} model where a small amount (typically $O(\lg n)$ bits) of clean space is provided along with some large additional auxiliary space, with the condition that the additional space is initially in an arbitrary, possibly incompressible, state and must be returned to this state when the computation is finished. The input is assumed to be given in ROM. Thus this model can be thought of as having an auxiliary storage that needs to be `restored' in contrast to the model by Chan et al.~\cite{ChanMR14} where the input array has to be `restored'. They show various interesting complexity theoretic consequences in this model and designed significantly better (in terms of space) algorithms in comparison with the ROM model for a few combinatorial problems. 
\end{itemize}

\subsection{Previous work in space efficient graph algorithms}
Even though these models were introduced in the literature with the aim of designing and/or implementing various algorithms space efficiently, {\it space efficient graph algorithms} have been designed only in the (semi)-streaming and the ROM model. In the streaming and semi-streaming models, researchers have studied several basic and fundamental algorthmic problems such as connectivity, minimum spanning tree, matching. See~\cite{McGregor14} for a comprehensive survey in this field. Research on these two models (i.e., streaming and semi-streaming) is relatively new and has been going on for last decade or so whereas the study in ROM could be traced back to almost 40 years. In fact there is already a rich history of designing space efficient algorithms in the read-only memory model. The complexity class {\sf L} (also known as {\sf DLOGSPACE}) is the class containing decision problems that can be solved by a deterministic Turing machine using only logarithmic amount of work space for computation. There are several important algorithmic results~\cite{DattaLNTW09,ElberfeldJT10,ElberfeldK14,ElberfeldS16} for this class, the most celebrated being Reingold's method~\cite{Reingold08} for checking {\it st}-reachability in an undirected graph, i.e., to determine if there is a path between two given vertices $s$ and $t$. {\sf NL} is the non-deterministic analogue of {\sf L} and it is known that the {\it st}-reachability problem for {\it directed} graphs is {\sf NL}-complete (with respect to log space reductions). Using Savitch's algorithm \cite{AroraB}, this problem can be solved in $n^{O(\lg n)}$ time using $O(\lg ^2 n)$ bits of extra space. Savitch's algorithm is very space efficient but its running time is superpolynomial. 

Among the deterministic algorithms running in polynomial time for directed {\it st}-reachability, the most space efficient algorithm is due to Barnes et al.~\cite{BarnesBRS98} who gave a slightly sublinear space (using $n/2^{\Theta(\sqrt{\lg n})}$ bits) algorithm for this problem
%directed {\it st}-connectivity 
running in polynomial time. We know of no better polynomial time algorithm for this problem with better space bound. Moreover, the space used by this algorithm matches a lower bound on space for solving directed {\it st}-reachability on a restricted model of computation called Node Naming Jumping Automata on Graphs (NNJAG’s)~\cite{CookR80,EdmondsPA99}. This model was introduced especially for the study of directed {\it st}-reachability and most of the known sublinear space algorithms for this problem can be implemented on it. Thus, to design any polynomial time ROM algorithm taking space less than $n/2^{\Theta(\sqrt{\lg n})}$ bits requires radically new ideas. Recently there has been some improvement in the space bound for some special classes of graphs like planar and H-minor free graphs~\cite{AsanoKNW14,ChakrabortyPTVY14}. A drawback for all these algorithms using small space i.e., sublinear number of bits, is that their running time is often some polynomial of high degree. This is not surprising as Tompa~\cite{Tompa82} showed that for directed {\it st}-reachability, if the number of bits available is $o(n)$ then some natural algorithmic approaches to the problem require super-polynomial time.

Motivated by these impossibility results from complexity theory and inspired by the practical applications of these fundamental graph algorithms, recently there has been a surge of interest in improving the space complexity of the fundamental graph algorithms without paying too much penalty in the running time i.e., reducing the working space of the classical graph algorithms to $O(n)$ bits with little or no penalty in running time. Thus the goal is to design space-efficient yet reasonably time-efficient graph algorithms on the ROM model. Generally most of the classical linear time graph algorithms take $O(n)$ words or equivalently $O(n \lg n)$ bits of space. Towards this
%In sharp contrast, for space efficient algorithms for DFS in ROM, the landscape looks markedly different. The standard recursive implementation of DFS uses a stack (that could grow to an $O(n)$ size taking $O(n \lg n)$ bits) and a color array of $O(n)$ bits, to traverse the graph in $O(m+n)$ time. 
%As the stack can store up to $O(n)$ elements in worst case, the space taken for implementing DFS is $O(n \lg n)$ bits. 
recently Asano et al.~\cite{AsanoIKKOOSTU14} gave an $O(m \lg n)$ time algorithm using $O(n)$ bits, and another implementation taking $n+ o(n)$ bits, using Reingold's or Barnes et al's reachability algorithm, and hence have high polynomial running time. 
Later, time bound was improved to $O(m \lg \lg n)$ still using $O(n)$ bits in~\cite{ElmasryHK15}. For sparse graphs, the time bound is further improved in~\cite{BCR,ChakrabortyRS16} to optimal $O(m)$ using still $O(n)$ bits of space. Improving on the classical linear time implementation of BFS which uses $O(n \lg n)$ bits of space, recent space efficient algorithms~\cite{BCR,ElmasryHK15,HagerupK16} have resulted in a linear time algorithm using $n \lg 3+o(n)$ bits. We know of no algorithm for BFS using $n+o(n)$ bits and $O(m \lg^c n)$ (or even $O(mn)$) time for some constant $c$ in ROM. The only BFS algorithm taking sublinear space uses $n/2^{\Theta(\sqrt{\lg n})}$ bits~\cite{BarnesBRS98} and has a high polynomial runtime. A few other space efficient algorithms for fundamental graph problems like checking strong connectivity~\cite{ElmasryHK15}, biconnectivity and performing {\it st}-numbering~\cite{ChakrabortyRS16}, recognizing chordal and outerplanar graphs~\cite{ChakrabortyS17,KammerKL16} were also designed very recently.

\subsection{In-place model for graph algorithms}
In order to break these inherent space bound barriers and obtain reasonably time and space efficient graph algorithms, we want to relax the limitations of ROM. And the most natural and obvious candidate in this regard is the classical {\it in-place} model. Thus our main objective is to initiate a systematic study of efficient in-place (i.e., using $O(\lg n)$ bits of extra space) algorithms for graph problems. To the best of our knowledge, this has not been done in the literature before.
% that include BFS, DFS in directed and undirected graphs. 
Our first goal towards this is to properly define models for the in-place graph algorithms. As in the case of standard in-place model, we need to ensure that the graph (adjacency) structure remains intact throughout the algorithm. Let $G=(V,E)$ be the input graph with $n=|V|$, $m=|E|$, and assume that the vertex set $V$ of $G$ is the set $V=\{1,2,\cdots,n\}$.  
%We use $n$ and $m$ to denote the number of vertices and the number of edges, respectively, of the input graph $G$.
%Moreover, we assume that the vertex set of $G$ is the set $\{1,\ldots, n\}$. 
To describe these models, we assume that the input graph representation consists of two parts: (i) an array $V$ of length $n$, where $V[i]$ stores a pointer to the adjacency list of vertex $i$, and (ii) a list of singly linked lists, where the $i$-th list consists of a 
singly linked list containing all the neighbors of vertex $i$ with $V[i]$ pointing to the head of the list. In the ROM model, we assume that both these components 
%of the input representation 
cannot be modified. In our relaxed models, we assume that one of these components can be modified in a limited way. This gives rise to two different model which we define next.

\noindent
{\bf Implicit model:} The most natural analogue of in-place model allows
any two elements in the adjacency list of a vertex to be swapped (in constant time assuming that we have access to the nodes storing those elements in the singly linked list). The adjacency ``structure'' of 
the representation does not change; only the values stored can be swapped. (One may restrict this further to allow
only elements in adjacent nodes to be swapped. Most of our algorithms work with this restriction.) We call it the \implicit\ model inspired by the notion of {\it implicit data structures}~\cite{Munro86}. We introduce and develop algorithms for another relaxed model which we call the
\rotated\ model. 

\noindent
{\bf Rotate model:}
In this model, we assume that only the pointers stored in 
the array $V$ can be modified, that too in a limited way - to point to any node in the adjacency list, instead of always pointing to the first node. 
In space-efficient setting, since we do not have additional space to store a pointer to the beginning of the adjacency list explicitly, we assume that the second component of the graph representation consists of a
list of circular linked lists (instead of singly linked lists) -- i.e., the last node in the adjacency list of each vertex points
to the first node (instead of storing a null pointer). See the Figure $1$ to get a better visual description. 
%1(b) in Appendix~\ref{rotatepic}.
We call the element pointed to by the pointer as the front of the list, and a unit cost rotate operation changes the element pointed to by the pointer to the next element in the list.
%We use `one step' rotation in a circular list, that makes the ``head'' pointer point to its next element in the list,  
%as a primitive operation in a circular list. A full rotation in $v$'s circular list means visiting every neighbor of $v$ once starting with $u$, the first vertex in the list, and stopping when the pointer points to $u$ again i.e. $u$ becomes the first vertex again. 

\begin{figure}[ht]
\label{fig:adjlist}
 \begin{center}
 \includegraphics[scale=.8, keepaspectratio=true]{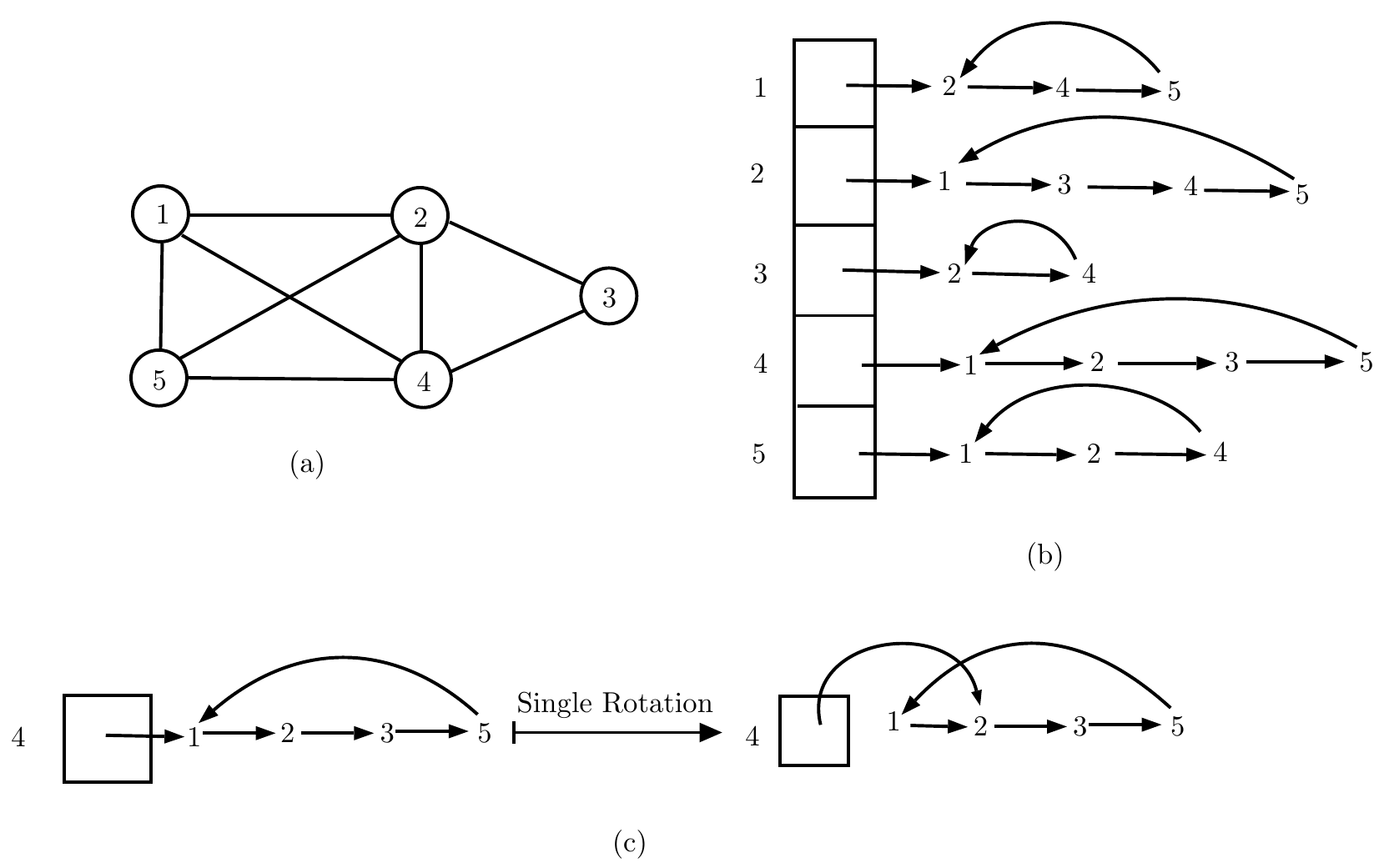}
 \end{center}
 \caption{(a) An undirected graph G with 5 vertices and 8 edges. (b) A circular list representation of G. To avoid cluttering
the picture, we draw the vertices and the pointers to the next node separately as opposed to a single node having two different fields in the circular list. (c) An illustration of a single clockwise rotation in the circular list of vertex 4.}
\end{figure}

Thus the \rotated\ model corresponds to keeping the adjacency lists in read-only memory and allowing (limited) updates on the pointer array that points to these lists. And, the \implicit\ model corresponds to the reverse case, where we keep the pointer array in read-only memory and allow swaps on the adjacency lists/arrays.
A third alternative especially for the \implicit\ model is to assume that the input graph is represented as an adjacency array, i.e., adjacency lists are stored as arrays instead of singly linked lists (see~\cite{ChakrabortyRS16,ElmasryHK15,KammerKL16} for some results using this model); and we allow here that any two elements in the adjacency array can be swapped. In this model, some of our algorithms have improved performance in time.
% \begin{figure}
% %\label{fig:adjlist}
%  \begin{center}
%  \includegraphics[scale=.6, keepaspectratio=true]{input}
%  \end{center}
%  \caption{(a) An undirected graph G with 5 vertices and 8 edges. (b) A circular list representation of G. To avoid cluttering
% the picture, we draw the vertices and the pointers to the next node separately as opposed to a single node having two different fields in the circular list. (c) An illustration of a single clockwise rotation in the circular list of vertex 4.}
% \end{figure}

\subsection{Definitions, computational complexity and notations}
We study some basic and fundamental graph problems in these models. In what follows we provide the definitions and state the computational complexity of some these problems. For the DFS problem, there have been two versions studied in the literature. In the {\it lexicographically smallest DFS} or {\it lex-DFS} problem, when DFS looks for an unvisited vertex to visit in an adjacency list, it picks the ``first'' unvisited vertex where the ``first'' is with respect to the appearance order in the adjacency list. The resulting DFS tree will be unique. In contrast to lex-DFS, an algorithm that outputs {\it some} DFS numbering of a given graph, treats an adjacency list as a set, ignoring the order of appearance of vertices in it, and outputs a vertex ordering $T$ such that there exists {\it some} adjacency ordering $R$ such that $T$ is the DFS numbering with respect to $R$. We say that such a DFS algorithm performs {\it general-DFS}. Reif~\cite{Reif85} has shown that lex-DFS is {\sf P}-complete (with respect to log-space reductions) implying 
that a logspace algorithm for lex-DFS results in the collapse of
complexity classes {\sf P} and {\sf L}. Anderson et al.~\cite{AndersonM87} have shown that even computing the leftmost root-to-leaf path of the lex-DFS tree is {\sf P}-complete. For many years, these results seemed to imply that the general-DFS problem, that is, the computation of any DFS tree is also inherently sequential. However, Aggarwal et al. \cite{AggarwalA88,AggarwalAK90} proved that the general-DFS problem can be solved much more efficiently, and it is in {\sf RNC}. Whether the general-DFS problem is in {\sf NC} is still open.

As is standard in the design of space-efficient algorithms~\cite{BCR,ElmasryHK15}, while working with directed graphs, we assume that the graphs are given as in/out (circular) adjacency lists i.e., for a vertex $v$, we have the (circular) lists of both in-neighbors and out-neighbors of $v$. We assume the word RAM model of computation where the machine consists of words of size $w$ in $\Omega(\lg n)$ bits and any logical, arithmetic and bitwise operation involving a constant number of words takes $O(1)$ time. We count space in terms of number of {\it extra} bits used by the algorithm other than the input, and this quantity is referred as ``extra space'' and ``space'' interchangeably throughout the paper. By a path of length $d$, we mean a simple path on $d$ edges. By $deg(x)$ we mean the degree of the vertex $x$. In directed graphs, it should be clear from the context whether that denotes out-degree or in-degree.
%
% As is standard in the design of recent space-efficient algorithm \cite{ElmasryHK15,BCR}, when we work with directed graphs, we assume that the graphs are represented as in/out (circular) adjacency lists i.e., given a vertex $v$, we have both the (circular) lists of in-neighbors and out-neighbors 
% of $v$.  Also 
By a BFS/DFS traversal of the input graph $G$, as in~\cite{AsanoIKKOOSTU14,BCR,ChakrabortyRS16,ElmasryHK15,HagerupK16}, we refer to reporting the vertices of $G$ in the BFS/DFS ordering, i.e., in the order in which the vertices are visited for the first time.

\subsection{Our Results}

\paragraph{Rotate Model:}
%\noindent
%{\bf Depth-first Search.}
For DFS, in the \rotated\ model, we show the following in Sections~\ref{dfs1},~\ref{dfs2} and~\ref{subsec:dfsrotate-logn}.
%Sections~\ref{subsec:rotatedfs}.

\begin{theorem}\label{dfsrotate}
Let $G$ be a directed or an undirected graph, and $\ell \leq n$ be the maximum depth of the DFS tree starting at a source vertex $s$. Then in the \rotated\ model, the vertices of $G$ can be output in
\begin{enumerate}%[noitemsep,topsep=0pt]
\item[$(a)$] the lex-DFS order in $O(m+n)$ time using $n \lg 3+O(\lg^2 n)$ bits,
\item[$(b)$] a general-DFS order in $O(m+n)$ time using $n+O(\lg n)$ bits, and 
\item[$(c)$] a general-DFS order in $O(m^2/n+m\ell)$ time for an undirected graph and in $O(m(n+\ell^2))$ time for directed graphs using $O(\lg n)$ bits. For this algorithm, we assume that $s$ can reach all other vertices.
\end{enumerate}
\end{theorem}

This is followed by the BFS algorithms where, in the \rotated\ model, we show the following in Sections~\ref{subsec:rotatebfs} and~\ref{subsec:rotatebfs-logn}.

\begin{theorem}\label{bfsrotate}
Let $G$ be a directed or an undirected graph, and $\ell$ be the depth of the BFS tree starting at the source vertex $s$. Then in the \rotated\ model, the vertices of $G$ can be output in a BFS order in 
\begin{enumerate}%[noitemsep,topsep=0pt]
\item[$(a)$] $O(m+n\ell^2)$ time using $n + O(\lg n)$ bits, and
\item[$(b)$] $O(m \ell + n \ell^2)$ time using $O(\lg n)$ bits. Here we assume that the source vertex $s$ can reach all other vertices.
\end{enumerate}
\end{theorem}

\paragraph{Implicit Model:} In the \implicit\ model, we obtain polynomial time implementations for lex-DFS and general-DFS using $O(\lg n)$ bits. For lex-DFS, this is conjectured to be unlikely in ROM as the problem is {\sf P}-complete~\cite{Reif85}. In particular, we show the following in Section~\ref{subsec:implicitdfs}.

\begin{theorem}\label{dfsimplicit}
Let $G$ be a directed or an undirected graph with a source vertex $s$ and $\ell\leq n$ be the maximum depth of the DFS tree starting at $s$ that can reach all other vertices. Then in the \implicit\ model, using $O(\lg n)$ bits the vertices of $G$ can be output in
\begin{enumerate}%[noitemsep,topsep=0pt]
\item[$(a)$]
the lex-DFS order in $O(m^3/n^2 + \ell m^2/n)$ time if $G$ is given in adjacency list and in $O(m^2 \lg n / n)$ time if $G$ is given in adjacency array for undirected graphs. For directed graphs our algorithm takes $O(m^2(n+\ell^2)/n)$ time if $G$ is given in adjacency list and $O(m\lg n (n+\ell^2))$ time if $G$ is given in adjacency array;

\item[$(b)$] a general-DFS traversal order in $O(m^2/n)$ time if the input graph $G$ is given in an adjacency list and in $O(m^2 (\lg n) / n + m \ell \lg n))$ time if it is given in an adjacency array.
%\begin{enumerate}
%\item
%\item
%in $O(m \lg m + n)$ time if the input is given in the form of an adjacency array.
%\end{enumerate}
\end{enumerate}
\end{theorem}

In the \implicit\ model, we can match the runtime of BFS from \rotated\ model, and do better in some special classes of graphs. In particular, we show the following in Section~\ref{omit_proof} .

\begin{theorem}\label{bfsimplicit}
Let $G$ be a directed or an undirected graph with a source vertex that can reach all other vertices by a distance of at most $\ell$.  Then in the \implicit\ model,  using $O(\lg n)$ bits the vertices of $G$ can be output in a BFS order in 
\begin{enumerate}%[noitemsep,topsep=0pt]
\item[$(a)$] $O(m + n \ell^2)$ time;
\item[$(b)$] the runtime can be improved to $O(m+n\ell)$ time if there are no degree $2$ vertices;
\item[$(c)$] the runtime can be improved to $O(m)$ if the degree of every vertex is at least $2 \lg n +3$.
\end{enumerate}
\end{theorem}

In sharp contrast, for space efficient algorithms for DFS in ROM, the landscape looks markedly different. To the best of our knowledge, there are no DFS algorithms in general graphs in ROM that use $O(n^{1-\epsilon})$ bits. In fact, an implementation of DFS taking $cn$ bits for $c<1$ has been proposed as an open problem by Asano et al.~\cite{AsanoIKKOOSTU14}. Similar to DFS, to the best of our knowledge, there are no polynomial time BFS algorithms in ROM that use even $O(n^{1-\epsilon})$ bits. On the other hand, we don't hope to have a BFS algorithm (for both undirected and directed graphs) using $O(\lg n)$ bits in ROM as the 
problem is {\sf NL}-complete~\cite{AroraB}. \\

{\bf Minimum Spanning Tree (MST).} Moving on from DFS and BFS, we also study the problem of reporting a minimum spanning tree (MST) of a given undirected connected graph $G$. We show the following result in Section~\ref{mstproof}.

\begin{theorem}\label{mst_proof}
A minimum spanning forest of a given undirected weighted graph $G$
%, where the weights of any edge can be represented in $O(\lg n)$ bits, 
can be found using $O(\lg n)$ bits and in
\begin{enumerate}
\item[$(a)$]  $O(mn)$ time in the \rotated\ model,
\item[$(b)$]  $O(mn^2)$ time in the \implicit\ model if $G$ is given in an adjacency list, and
\item[$(b)$]  $O(mn \lg n)$ time in the \implicit\ model when $G$ is represented in an adjacency array.
\end{enumerate}
\end{theorem}

Note that by the results of~\cite{Reif84,Reingold08}, we already know logspace algorithms for MST in ROM but again the drawback of those algorithms is their large time complexity. On the other hand, our algorithms have relatively small polynomial running time, simplicity, making it an appealing choice in applications with strict space constraints.

\subsection{Techniques}
%Our implementations follow (variations of) the classical algorithms for BFS and DFS that use three colors (\white, \gray\ and \black), but avoid the use of stack (for DFS) and queue (for BFS). In the \rotated\ model, when a node is visited for the first time, we use the rotate operation to move the parent or a (typically the currently explored) child to the beginning of the list to help navigate through the tree during the forward or the backtracking step. This avoids the need for stack or queue. Algorithms using $O(\lg n)$ bits (which don't even have space to store the color array) use the rotate operation in a non-trivial way to move elements within the lists to determine the color of the vertices as well.

Our implementations follow (variations of) the classical algorithms for BFS and DFS that use three colors (\white, \gray\ and \black), but avoid the use of stack (for DFS) and queue (for BFS). In the \rotated\ model, we first observe that in the usual search algorithms one can dispense with the extra data structure space of pointers maintaining the search tree (while retaining the linear number of bits and a single bit per vertex in place of the full unvisited/visited/explored array) simply by rotating each circular adjacency lists to move the parent or a (typically the currently explored) child to the beginning of the list to help navigate through the tree during the forward or the backtracking step, i.e. by changing the pointer from the vertex to the list of its adjacencies by one node at a time. This retains the basic efficiency of the search strategies. The nice part of this strategy is that the total number of rotations also can be bounded. To reduce the extra space from linear to logarithmic, it is noted that one can follow the vertices based on the current rotations at each vertex to determine the visited status of a vertex, i.e. these algorithms use the rotate operation in a non-trivial way to move elements within the lists to determine the color of the vertices as well. However, the drawback is that to do so could require moving up (or down) the full height of the implicit search tree.  This yields super-linear rather than (near-) linear time algorithms.

In the \implicit\ model, we use the classical {\it bit encoding} trick used in the development of the implicit data structures~\cite{Munro86}.
We encode one (or two) bit(s) using a sequence of two (or three respectively) distinct numbers. 
To encode a single bit $b$ using two distinct values $x$ and $y$ with $x < y$, we store the sequence 
$x, y$ if $b = 0$, and $y, x$ otherwise.
Similarly, permuting three distinct values $x, y, z$ with $x < y < z$, we can represent six combinations. We can choose any of the four combinations to represent up to $4$ colors (i.e. two bits). 
Generalizing this further, we can encode a pointer taking $\lg n$ bits using $2 \lg n$ distinct elements where reading or updating a bit takes constant time, and reading or updating a pointer takes $O(\lg n)$ time. This also is the reason for the requirement of vertices with (high) degree at least $3$ or $2 \lg n +3$ for faster algorithms, which will become clear in the description of the algorithms.

%\subsection{Simulations across models}
%In general, an algorithm implemented in the \rotated\ model can be implemented 
%\begin{itemize}%[noitemsep,topsep=0pt]
%\item
%in the \implicit\ model with a slow-down of a factor of at most $n$ in the runtime, as a rotate operation can be implemented in time proportional to the (maximum) degree of a vertex; 
%\item
%in read-only memory using an extra $O(n \lg (m/n))$ bits by storing a pointer in each of the arrays in the adjacency array representation, which can be updated in constant time.
%\end{itemize}
%These are discussed in detail in Section~\ref{subsec:Rotate-Implicit} and in Appendix~\ref{romdfs}.
%Similarly any algorithm in ROM can be implemented in the \rotated\ and \implicit\ models without any additional time or space. The \implicit\ model is the most general model and not surprisingly it is not easy to simulate algorithms of that model, in general, in ROM or in \rotated\ models without substantial loss of time and/or space.

\subsection{Consequences of our BFS and DFS results}
There are many interesting and surprising consequences of our results for BFS and DFS in both the \rotated\ and \implicit\ model. In what follows, we mention a few of them.
\begin{itemize}
\item For {\it directed st-reachability}, as mentioned previously, the most space efficient polynomial time algorithm~\cite{BarnesBRS98} uses $n/2^{\Theta(\sqrt{\lg n})}$ bits. In sharp contrast, we obtain efficient (timewise) log-space algorithms for this problem in both the \rotated\ and \implicit\ models (as a corollary of our directed graph DFS/BFS results). In terms of workspace this is exponentially better than the best known polynomial time algorithm~\cite{BarnesBRS98} for this problem in ROM. For us, this provides one of the main motivations to study this model. A somewhat incomparable result obtained recently by Buhrman et al.~\cite{BuhrmanCKLS14,Koucky16} where they designed an algorithm for {\it directed st-reachability} on catalytic Turing machines in space $O(\lg n)$ with catalytic space $O(n^2 \lg n)$ and time $O(n^9)$.

\item Problems like {\it directed st-reachability}~\cite{AroraB}, {\it distance}~\cite{Tantau07} which asks whether a given $G$ (directed, undirected or even directed acyclic) contains a path of length at most $k$ from $s$ to $t$, are {\sf NL}-complete i.e., no deterministic log-space algorithm is known for these problems. But in our (both the  \rotated\ and \implicit) models, we design log-space algorithms for them. Assuming {\sf L} $\neq$ {\sf NL}, these results show that probably both our models with log-space are stronger than {\sf NL}. 

\item The lex-DFS problem (both in undirected and directed graphs) is {\sf P}-complete~\cite{Reif85}, and thus polylogarithmic space algorithms are unlikely to exist in the ROM model. But we show an $O(\lg n)$ space algorithm in the \implicit\ model for lex-DFS. This implies that, probably the \implicit\ model is even more powerful than the  \rotated\ model. It could even be possible that every problem in {\sf P} can be computed using log-space in the \implicit\ model. A result of somewhat similar flavor is obtained recently Buhrman et al.~\cite{BuhrmanCKLS14,Koucky16} where they showed that any function in ${\sf TC^1}$ can be computed using catalytic log-space, i.e., ${\sf TC^1} \subseteq {\sf CSPACE(\lg n)}$. Note that ${\sf TC^1}$ contains {\sf L}, {\sf NL} and even other classes that are conjectured to be different from {\sf L}. 
% We refer the readers to their paper for details. 

\item Our bounds for BFS and DFS in the \rotated\ and \implicit\ models immediately imply (with some care) similar bounds, that are improvement over the best space bounds known so far in ROM, for many applications of DFS/BFS. Moreover, as described before, any algorithm in the \rotated\ model can be implemented in ROM using extra $O(n \lg (m/n))$ bits. Thus, our linear time DFS algorithm in the \rotated\ model can be implemented in ROM using $O(n \lg (m/n))$ bits, matching the bound 
of~\cite{ChakrabortyRS16} for DFS. Using this DFS implementation, we can obtain improved space efficient algorithms for various applications of DFS in ROM. This is discussed in Section~\ref{romdfs}.

\item In Section~\ref{baker} we present Logspace Approximation Scheme or LSAS ($(1\pm \epsilon)$ approximation algorithm running in logspace for every $\epsilon> 0$) for a class of MSO-definable optimization problems which are amenable to the Baker's method~\cite{Baker94} in locally bounded treewidth graphs in both of our \rotated\ and \implicit\ models. No such algorithms are known in ROM as Baker's method requires to compute {\it distance} which is {\sf NL}-complete. As BFS admits logspace algorithms in our models, we can design such LSAS for these problems here.

\item For a large number of NP-hard graph problems, the best algorithms in ROM run in exponential time and polynomial space. We show that using just logarithmic amount of space, albeit using exponential time, we can design algorithms for those NP-hard problems in both of our models under some restrictions. This gives an exponential improvement over the ROM space bounds for these problems. This is described in Section~\ref{nphardlogspace}. In constrast, note that, no NP-hard problem can be solved in the ROM model using $O(\lg n)$ bits unless {\sf P=NP}. 
\end{itemize}

\section{DFS algorithms in the \rotated\ model}
In this section, we describe our space-efficient algorithms for DFS in the \rotated\ model proving Theorem~\ref{dfsrotate}.

\subsection{Proof of Theorem~\ref{dfsrotate}(a) for undirected graphs}\label{dfs1}
%Lex-DFS in $O(m+n)$ time and $n \lg 3+O(\lg^2 n)$ bits}
We begin by describing our algorithm for 
 undirected graphs, and later mention the changes required for directed graphs.
In the normal exploration of DFS (see for example, Cormen et al.~\cite{CLRS}) we use three colors. Every vertex $v$ is \white\ initially while it has not been discovered yet, becomes \gray\ when DFS discovers $v$ for the first time, and is colored \black\ when it is finished i.e., all its neighbors have been explored completely.
% We say, for a vertex $v$, the discovery time of $v$ is the time when $v$ changes its color from \white\ to \gray\ and the finishing time is the time when $v$ changes its color from \gray\ to \black. The standard implementation of DFS uses a stack to keep track of the vertices whose processing has begun but not completed yet. 

% In the space efficient implementation of DFS, following Asano et al.~\cite{AsanoIKKOOSTU14},
% DFS traversal can be thought of as doing the following two steps repeatedly. Let $u$ be any node (initially the root) of the partially explored DFS tree.
% First step takes place when DFS discovers $u$'s neighbor $v$ for the first time, and as a result $v$'s color changes from \white\ to \gray. We name this phase as the \emph{forward} step. When DFS completes exploring $v$ i.e., the subtree rooted at $v$ in the DFS tree, it needs to perform two tasks subsequently. First, it backtracks to its parent $u$, and then looks for the next \white\ neighbor to explore in $u$'s list. The latter part is almost similar to the forward step discussed before, and so we call the first part alone as the \emph{backtracking} step. 
% It is clear that if the algorithm keeps repeating these steps, it correctly visits all the vertices of $G$ eventually in a depth first manner. 
%We describe how to implement each step in detail using a color array that uses $3$ colors, 
% and later we improve it to use only $2$ colors. 

We maintain a color array $C$ of length $n$ that stores the color of each vertex at any point in the algorithm. In the rest of the paper, when we say we scan the adjacency list of some vertex $v$, what we mean is, we create a temporary pointer pointing to the current first element of the list and move this temporary pointer untill we find the desired element. Once we get that element we actually rotate the list so that the desired element now is at the front of the list. 
We start DFS at the starting vertex, say $s$, changing its color from \white\ to \gray\ in the color array $C$. Then we scan the adjacency list of $s$ to find the first \white\ neighbor, say $w$. We keep rotating the list to bring $w$ to the front of $s$'s adjacency list (as the one pointed to by the head $V[s]$), color $w$ \gray\ in the color array $C$ and proceed to the next step (i.e. to explore $w$'s adjacency list). This is the first \emph{forward} step of the algorithm. 
In general, at any step during the execution of the algorithm, whenever we arrive at a \gray\ vertex $u$ (including the case when $u$'s color is changed from \white\ to \gray\ in the current step), we scan $u$'s adjacency list to find the first \white\ vertex.
(i) If we find such a vertex, say $v$, then we rotate $u$'s list to make $v$ as the first element, and change the color of $v$ to \gray.
% This concludes the description of the forward step. 
% {\it We maintain the invariant that for every gray vertex, the first vertex in its adjacency list is its last child in the (current) partially explored DFS tree.}
(ii) If we do not find any \white\ vertex, then we change the color of $u$ to \black, and {\em backtrack} to its parent.
%
%To determine whether a grayed vertex is \black, we need to figure out that its adjacency list has no \white\ neighbors. As the adjacency list is a circular list, it is easy to determine that it has no \white\ neighbors by remembering its first vertex while scanning (and rotating) the list until we hit the first vertex again. 
%At this point (when a node becomes \black\), we color it \black\ and we need to backtrack to its parent. 
To identify $u$'s parent, we use the following lemma. 

%Proof is in Appendix \ref{l1}.

\begin{lemma}\label{backtrack}
Suppose $w$ is a node that just became \black. Then its parent $p$ is the 
unique vertex in $w$'s adjacency list which is (a) \gray\ and (b) whose 
current adjacency list has $w$ in the first position.
\end{lemma}

\begin{proof}
Among all the neighbors of $w$, some vertices are $w$'s children in the DFS tree, 
and the rest of them are $w$'s ancestors, and among the ancestors, exactly one 
vertex is $w$'s parent in the DFS tree. All the ancestors should have their currently 
explored (gray) child at the first position in their adjacency list; and this current child 
would be different from $w$ for all the ancestors except $p$ (as $w$ was 
discovered from $p$). So, the second condition is violated for them. All of $w$'s 
children have been fully processed earlier and have been colored \black, and hence 
the first condition is violated for them. Observe that, if $w$ has a child, say $k$, which is a leaf in the DFS tree, it might 
happen that $k$ also has $w$ at the first position in its current adjacency list, but, 
fortunately, $k$ is \black\ while scanning $w$'s list. So for such vertices, the first 
condition gets violated. Only for $w$'s parent, which is $p$ here, both the 
conditions are satisfied. 
\end{proof}

So, the parent can be found by by scanning the $w$'s list, to find a 
neighbor $p$ that is colored \gray\ such that the first element in $p$'s list is $u$. 
This completes the description of the backtracking step. 
Once we backtrack to $p$, we find the next \white\ vertex (as in the forward step) 
and continue until all the vertices of $G$ are explored. Other than some constant number of variables, clearly the space usage is only for storing the color array $C$.
% and $C$ can be implemented as described in Lemma \ref{nlgc}.
Since $C$ is of length $n$ where each element has $3$ possible values, $C$ 
can be encoded using $n \lg 3 + O(\lg^2 n)$ bits, so that the $i$-th element in $C$
can be read and updated in $O(1)$ time~\cite{DodisPT10}.
So overall space required is $n \lg 3 + O(\lg^2 n)$ bits. 
%In order to show that the algorithm runs in linear time, we observe that we do not scan or
%rotate the adjacency list of any vertex ``too often''. The following lemma is immediate, 
As the algorithm systematically brings a \white\ vertex to the front, makes it \gray, and 
moves it to the end after it becomes \black, at most two full rotations of each of the list may happen
(the second one to determine that there are no more \white\ vertices) resulting in a linear time lex-DFS algorithm. We discuss the lex-DFS algorithm for the directed graphs below.

\subsection{Proof of Theorem~\ref{dfsrotate}(a) for directed graphs}\label{d1}
Recall that we have access to both the in-adjacency and the out-adjacency lists for each vertex $w$ in a directed graph $G$, hence we can use these lists separately for performing two steps of DFS. I.e., out-adjacency list is used for the exploration of DFS in the forward direction and the in-adjacency list is used for finding parent of a node during the backtracking step. We provide the details below. Similar to our algorithm for undirected graphs, in the forward direction, we scan the out-neighbor list of $w$ to find the next \white\ neighbor and proceed. Once the out-neighbor list of $w$ is fully processed, we need to backtrack from $w$. Towards that we first have to identify $w$'s parent. In order to do so we use the following lemma whose proof follows along the same lines as the Lemma~\ref{backtrack} above. Hence we omit the proof.

\begin{lemma}
Suppose $w$ is a node that just became \black. Then its parent $p$ is the 
unique vertex in $w$'s in-adjacency list which is (a) gray and (b) whose 
current out-adjacency list has $w$ in the first position.
\end{lemma}

Once we figure out $w$'s parent $p$, DFS backtracks to $p$, finds the next \white\ neighbor (as done in the forward step) and continues until all the vertices are exhausted. It is clear that this procedure performs lex-DFS on a directed graph $G$ correctly in linear time, and this completes the proof of Theorem~\ref{dfsrotate}(a).

\subsection{Proof of Theorem~\ref{dfsrotate}(b) for undirected graphs}\label{dfs2}
%DFS in $O(m+n)$ time and $n+O(\lg n)$ bits}
%%\subsection{Improving the space to $n+O(\lg n)$ bits}\label{dfs2}
% We improve the space further to $n+O(\lg n)$ bits by removing one extra color. Note that, we crucially used all the three colors in especially in Lemma \ref{backtrack} to correctly figure out the parent of a vertex while backtracking. In what follows, we show how to achieve the same effect using just $2$ colors. 
To improve the space further, we replace the color array $C$ with a bit array $visited[1,\dots,n]$ which stores a $0$ for an unvisited vertex (\white ), and a $1$ for a visited vertex (\gray\ or \black ). 
% Thus the bit $0$ (in the visited array) is used to indicate the color \white, but we overload the bit $1$ to indicate either color \gray\ or \black. 
%

First we need a test similar to that in the statement of Lemma~\ref{backtrack} without the distinction of \gray\ and \black\ vertices to find the parent of a node. 
Due to the invariant we have maintained, every internal vertex of the DFS tree will point to (i.e. have as first element in its list) its last child. So the nodes that could potentially have a node $w$ in its first position are its parent, and {\it any} leaf vertex. Hence we modify the forward step in the following way.
%make the following modification in the forward step.

%First we implement the forward step as before in the previous algorithm i.e., while scanning vertex $v$'s adjacency list, keep rotating until the algorithm finds an unvisited vertex $u$, change $u$'s status to visited, and start processing $u$'s adjacency list. We will modify this a bit later. As for correctly implementing the backtracking step, suppose we want to backtrack from a vertex $w$ which has just turned \black. 
%Using visited/unvisited terminology, all the ancestors of $w$ are visited already, and they (except $w$'s parent) have their currently explored child, which is different from $w$, in the first location in their corresponding list. Hence, while backtracking, these ancestors can be ruled out easily. 
%But, as mentioned earlier in the remark after Lemma \ref{backtrack}, for some descendant of $w$, say $k$ (which has been completely explored, and hence \black), it is possible that $k$'s adjacency list contains $w$ at the first position. 
%%And, $k$ is also marked as visited while backtracking from $w$. 
%Thus there could be more than one vertex satisfying the property of Lemma~\ref{backtrack}.

%To deal with this, we change the forward step slightly. 
Whenever we visit an unvisited vertex $v$ for the first time from another vertex $u$ (hence, $u$ is the parent of $v$ in the DFS tree and $u$'s list has $v$ in the first position), we, as before, mark $v$ as visited and in addition to that, we rotate $v$'s list to bring $u$ to the front (during this rotation, we do not mark any intermediate nodes as visited). 
Then we continue as before (by finding the first unvisited vertex and bringing it to the front) in the forward step. 
Now the following invariants are easy to see and are useful. \\
\noindent
{\bf Invariants:} During the exploration of DFS, in the (partial) DFS tree
\begin{enumerate}[noitemsep,topsep=0pt]
\item
any internal vertex has the first element in its list as its current last child; and
\item
%if a vertex has been discovered, and completely explored (i.e. in the normal coloring procedure, its color is \black\), then the vertex is marked $1$ in the visited array, and
for any leaf vertex of the DFS tree, the first element in its list is its parent.
\end{enumerate}
The first invariant is easy to see as we always keep the current explored vertex (child) as the first element in the list. For leaves, the first time we encounter them, we make its parent as the first element in the forward direction. Then we discover that it has no unvisited vertices in its list, and so we make a full rotation and bring the parent to the front again. The following lemma provides a test to find the parent of a node.
 
\begin{lemma}\label{backtrack2}
Let $w$ be a node that has just become \black. Then its parent $p$ is the {\bf first} vertex $x$ in $w$'s list which is marked $1$ in the visited array, and whose current adjacency list has $w$ in the first position.
\end{lemma}
 \begin{proof}
 From the invariants we observed, the nodes that can potentially have $w$ in the first position of their lists are its parent and its children that happen to be leaves. But in $w$'s list, as we began the exploration of its neighbors starting from its parent, its parent will appear first before its children. Hence the first node in $w$'s list which has $w$ in the first position must be its parent.
 \end{proof}
 
Once we backtrack to $p$, we find the next \white\ vertex, and continue until all the vertices of $G$ are explored. Overall this procedure takes linear time. As we rotate the list to bring the parent of a node, before exploring its \white\ neighbors, we are not guaranteed to explore the first \white\ vertex in its original list, and hence we loose the lexicographic property. We provide our DFS algorithm for directed graphs below.

\subsection{Proof of Theorem~\ref{dfsrotate}(b) for directed graphs}\label{d2}
For performing DFS in directed graphs using  $n+O(\lg n)$ bits, we don't even need to apply the modifications as we did for the undirected graphs during the forward direction, and we can essentially use the same forward step idea as used for lex-DFS in undirected graphs of Section~\ref{dfs1}. We provide the details below. When we arrive at a previously unvisited vertex $v$ from the vertex $u$ (hence $u$ is the parent of $v$ in the DFS tree), we rotate the in-neighbor list of $v$ to bring $u$ to the front and $u$ stays there during the entire course of the exploration. Thus we maintain the invariant that for any visited node $v$, the first element in its in-neighbor list is its parent in the DFS tree. Now the algorithm scans $v$'s adjacency list to find its unvisited neighbor. (i) If we find such a vertex, say $w$, then we rotate $v$'s list to make $w$ as the first element, and mark $w$ visited.
(ii) If we do not find any such unvisited neighbor of $v$, then DFS needs to {\em backtrack} to its parent. From the invariant we maintain in the in-neighbor list of every visited vertex, this is easy. All we need to do is to see the first entry in $v$'s in-neighbor list to retrieve its parent $u$ and then continue from $u$. Overall this procedure takes linear time. Upon closer inspection, it can be seen that, as we are not modifying the ordering of the vertices in the out-neighbor lists in the forward direction (in contrast with the undirected graph algorithm of Section~\ref{dfs2}), this procedure actually traverses the directed graph $G$ in lex-DFS ordering. This completes the proof of Theorem~\ref{dfsrotate}(b).

\subsection{Proof of Theorem~\ref{dfsrotate}(c) for undirected graphs}\label{subsec:dfsrotate-logn}
%DFS using $O(\lg n)$ bits}\label{subsec:dfsrotate-logn}
Now to decrease the space to $O(\lg n)$, we dispense with the color/visited array, and give tests to determine \white, \gray\ and \black\ vertices. For now, assume that we can determine the color of a vertex. 
% We now describe an algorithm to perform a DFS traversal of $G$ using $O(\lg n)$ bits in the \rotated\ model. 
% We start off by describing our algorithm for undirected graphs first and later mention the changes required to make it work for directed graphs. As we don't have any extra space, 
%The challenge is to figure out the color of each vertex without the color array. 
% At each step we essentially have to test whether a vertex $v$ is either \white, \black\ or \gray. 
%Assume for now that we can test for visited vertices. Then
The forward step is almost the same as before except performing the update in the color array. I.e., whenever we visit a \white\ vertex $v$ for the first time from another vertex $u$ (hence $u$ is the parent of $v$), we rotate $v$’s list to bring $u$ to the front. Then we continue to find the first \white\ vertex to explore. We maintain the following invariants. (i) any \gray\ vertex has the first element in its list as its last child in the (partial) DFS tree; (ii) any \black\ vertex has its parent as the first element in its list. We also store the depth of the current node in a variable $d$, which is incremented by $1$ every time we discover a \white\ vertex and decremented by $1$ whenever we backtrack. We maintain the maximum depth the DFS has attained using a variable $max$. At a generic step during the execution of the algorithm, assume that we are at a vertex $x$'s list, 
let $p$ be $x$'s parent and let $y$ be a vertex in $x$'s list. We need to determine the color of $y$ and continue the DFS based on the color of $y$. We use the following characterization. 
%Proof is in Appendix \ref{l3}.

\begin{lemma}\label{backtrack3}
Suppose the DFS has explored starting from a source vertex $s$, up to a vertex $x$ at level $d$. Let $p$ be $x$'s parent.
Note that both $s$ and $x$ are \gray\ in the normal coloring procedure.
Let $max$ be the maximum level of any vertex in the partial DFS exploration. Let $y$ be a vertex in $x$'s list. Then,
\begin{enumerate}[noitemsep,topsep=0pt]
\item
$y$ is \gray\ (i.e., $(x,y)$ is a back edge, and $y$ is an ancestor of $x$) if and only if we can reach $y$ from $s$ by following through the \gray\ child 
(which is at the front of a \gray\ node's list) path in at most $d$ steps.
\item
$y$ is \black\ (i.e., $(x,y)$ is a back edge, and $x$ is an ancestor of $y$) if and only if
\begin{itemize}%[noitemsep,topsep=0pt]
\item
 there is a path $P$ of length at most $(max - d)$ from $y$ to $x$ (obtained by following through the first elements of the lists of every vertex in the path, starting from $y$), and
\item
let $z$ be the node before $x$ in the path $P$. The node $z$ appears after $p$ in $x$'s list.
\end{itemize}
\item
$y$ is \white\ if $y$ is not \gray\ or \black.
\end{enumerate}
\end{lemma}

 \begin{proof}
 The test for \gray\ and \white\ vertices is easy to see. The vertex $y$ is \black\ implies that $y$ is a descendant of $x$ in the partially explored DFS tree.
 This means that there is a path of length at most $(max - d)$ (obtained by following the parent which is in the first element of the adjacency list) from $x$ to $y$ through an already explored child $z$ . 
 By the way we process $x$'s list, we first bring the parent to the front of the list, and then explore the nodes in sequence, and hence $z$, the explored neighbor of $x$ must appear after $p$ in $x$'s list. Conversely, the unexplored neighbors of $x$ appear before $p$ in $x$'s list.
 \end{proof}

%The only thing left to describe is how we maintain the invariants for the \gray\ and \black\ vertices. By using the above claim, while exploring a \gray\ vertex $x$, if we find a \white\ vertex $y$, then we keep $y$ in the first position in $x$'s list. On the other hand, if all the neighbors of $x$ are non-\white, we bring the parent $p$ of $x$ to the front of $x$'s list and backtrack to $p$. Figuring out the parent can be done while testing for \gray\ (the last vertex before $x$ on the root to $x$ path is $x$'s parent). 
Now, if we use the above claim to test for colors of vertices, testing for \gray\ takes at most $d$ steps. Testing for \black\ takes at most $(max -d)$ steps to find the path, and at most $deg(x)$ steps to determine whether $p$ appears before. Thus for each vertex in $x$'s list, we spend time proportional to $max + deg(x)$. So, the overall runtime of the algorithm is
$\sum_{v\in V} deg(v)(deg(v)+\ell)=O(m^2/n+m\ell)$, where $\ell$ is the maximum depth of DFS tree. Maintaining the invariants for the \gray\ and \black\ vertices are also straightforward. We provide the details of our log-space algorithm for directed graphs below.
%in Appendix~\ref{d3}.

\subsection{Proof of Theorem~\ref{dfsrotate}(c) for directed graphs}\label{d3}
We describe our $O(\lg n)$ bits algorithm for directed graphs in the \rotated\ model. More specifically, we give a DFS algorithm 
to output all vertices reachable by a directed path from the source vertex $s$. If we assume that $s$ can reach all vertices, we get to output all vertices.
In the preprocessing step, the algorithm spends $O(m)$ time to bring the minimum valued neighbor (denote it by $min$) in the out-neighbor list of every vertex by rotation (hence we loose the lexicographic DFS property). For now assume that we can determine the color of a vertex. Given this, in the forward direction, when DFS arrives at a previously unvisited vertex $v$ from the vertex $u$ (hence $u$ is the parent of $v$ in the DFS tree), we rotate the in-neighbor list of $v$ to bring $u$ to the front and $u$ stays there during the entire course of the exploration. Also in $u$'s out-neighbor list, $v$ is made the first location. Hence we maintain the following invariants.\\
\noindent
{\bf Invariants:} During the exploration of DFS, in the (partial) DFS tree
\begin{enumerate}[noitemsep,topsep=0pt]
\item
\gray\ vertices have their current last child in the first location of their out-neighbor lists;
\item
%if a vertex has been discovered, and completely explored (i.e. in the normal coloring procedure, its color is \black\), then the vertex is marked $1$ in the visited array, and
all the visited (i.e., \gray\ and \black) vertices have their parent in the first location of their in-neighbor lists.
\end{enumerate}
We also keep track of the depth of the current node (i.e., the last \gray\ vertex in the \gray\ path of the DFS tree) in a variable $d$, which, as before, is incremented by $1$ every time DFS visits a \white\ vertex and decremented by $1$ whenever DFS backtracks. We also store the maximum depth the DFS tree has attained so far in a variable $max$. At a generic step during the execution of the algorithm, assume that we are at a vertex $x$'s list, let $p$ be $x$'s parent (which can be found from the way $x$ is visited by a forward or a backtracking step using the invariants being maintained) and let $y$ be a vertex in $x$'s list. 
We need to determine the color of $y$ and continue the DFS based on the color of $y$ and maintain the invariants. 
We use the following characterization.

\begin{figure}[h]
%\label{fig:colorchoices}
%\label{fig:adjlist}
 \begin{center}
 \includegraphics[scale=.6, keepaspectratio=true]{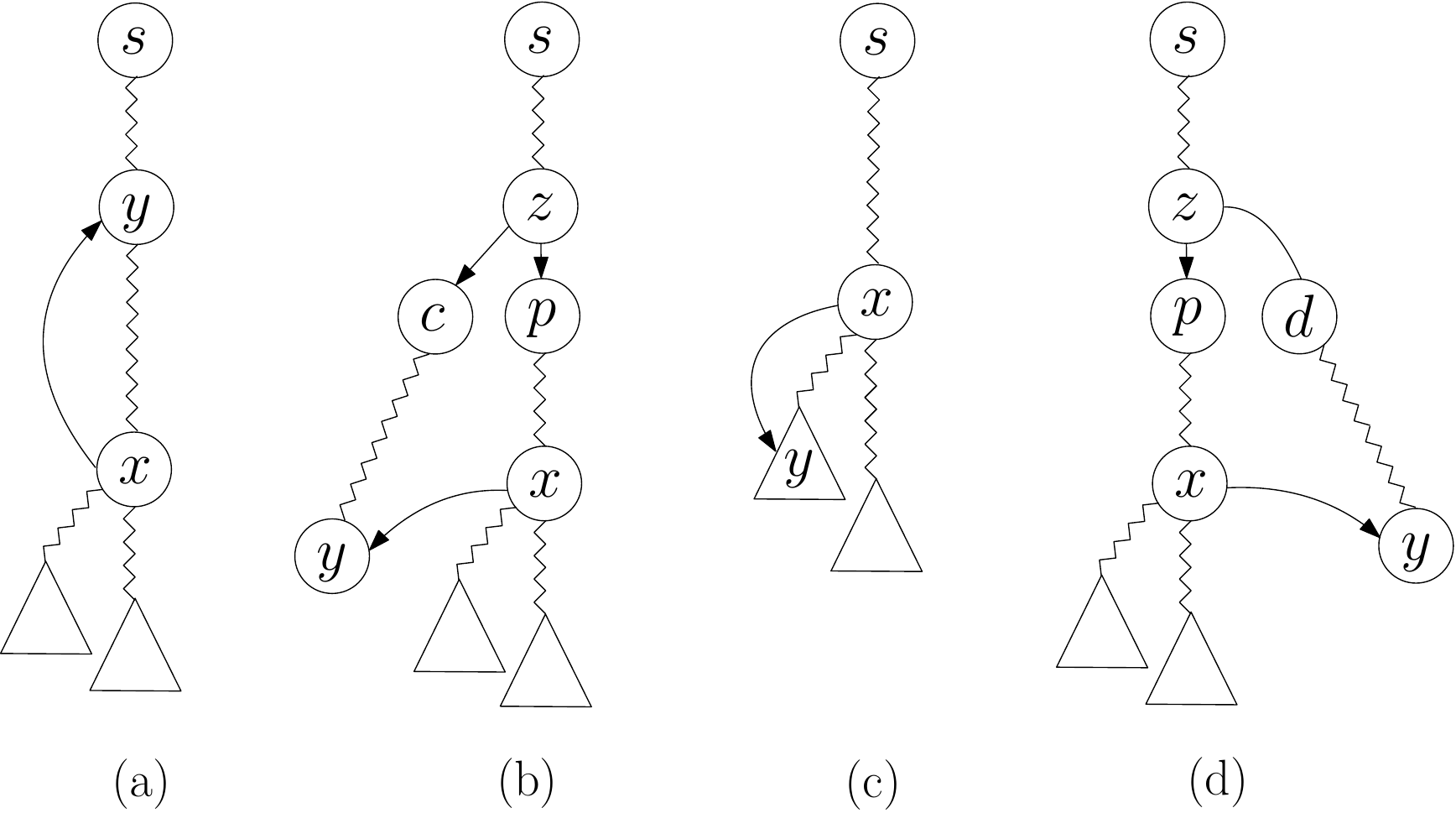}
 \end{center}
 \caption{Illustration of the different cases of the possible positions of the vertex $y$ when DFS considers the directed edge $(x,y)$ at some intermediate step. Suppose the root of the DFS tree is the vertex $s$ and the curvy path starting from $s$ and going straight below through $x$ is the current \gray\ path in the DFS tree. Intuitively all the vertices on the left hand side of the path are black, and right hand side are white and yet to be explored. From left to the right are cases when $(x,y)$ is (a) back edge, (b) cross edge, (c) forward edge, and (d) tree edge.}
\end{figure}

\begin{lemma}
Suppose the DFS has explored starting from a source vertex $s$ up to a vertex $x$ at level $d$. Let $p$ be $x$'s parent. Note that both $s$ and $x$ are \gray\ in the normal coloring procedure. Let $max$ be the maximum level of any vertex in the partial DFS exploration. Let $y$ be a vertex in $x$'s list. Then,
\begin{enumerate}[noitemsep,topsep=0pt]
\item
$y$ is \gray\ (i.e., $(x,y)$ is a back edge and $y$ is an ancestor of $x$) if and only if we can reach $y$ from $s$ following the \gray\ child (which is in the first location of each of the out-neighbor lists of \gray\ nodes) path in at most $d$ steps.
\item
$y$ is \black\ 
%(i.e. $(x,y)$ is a back edge, and $x$ is an ancestor of $y$) 
if and only if any of the following happens.
\begin{itemize}%[noitemsep,topsep=0pt]
\item
There is a path $P$ of length at most $(max - d)$ from $y$ to $x$ obtained by following the first elements of the in-neighbor lists of every vertex in the path $P$ starting from $y$. This happens when $(x,y)$ is a forward edge, and $x$ is an ancestor of $y$.
\item
%There is a path $P$ of length at most $max$ from $y$ to $z$ such that $z$ is an ancestor of $x$ having highest depth in the DFS tree and let $c$ be the node before $z$ in the path $P$, then $c$ must appear after $min$ in $z$'s list (this happens when $(x,y)$ is a cross edge).
There is a path $P$ of length at most $max$ from $y$ to a gray vertex $z \neq x$ (obtained by following through the first elements of the in-neighbor lists of every vertex starting from $y$) which is the first gray vertex in the path. 
% T$z$ is an ancestor of $x$ having highest depth in the DFS tree and 
Let $c$ be the node before $z$ in the path $P$, then $c$ must appear after $min$ in $z$'s list (this happens when $(x,y)$ is a cross edge).
%The node $z$ appears after $p$ in $x$'s list.
\end{itemize}
\item
The vertex $y$ is \white\ if it is not \black\ or \gray (i.e., $(x,y)$ is the next tree edge with $x$ being the parent of $y$ in the DFS tree). 
%$y$ is \white\ if $y$ is not \gray\ or \black.
\end{enumerate}
\end{lemma}

\begin{proof}
See Figure $2$
%~\ref{fig:colorchoices} 
for a picture of all the cases. The test for \gray\ and \white\ vertices is easy to see.  

From a vertex $x$, there could be two types of outgoing edges to a black vertex $y$. When $(x,y)$ is a forward edge, 
$y$ is a descendant of $x$ and hence there must exist a path $P$ of length at most $(max - d)$ (obtained by following the parent which is in the first location of the in-neighbor list of every vertex in $P$, starting from $y$) from $y$ to $x$ through an already explored child $t$ of $x$. In the other case, when $(x,y)$ is a cross edge, $y$ has already been discovered and explored completely before DFS reaches to $x$. Hence there must exist a \gray\ ancestor $z$ of $x$ ($z$ could be $x$) such that $y$ belongs to the subtree rooted at $z$ in the DFS tree. Thus, from $y$'s in-neighbor list if we follow the path starting with $y$'s parent for at most $max$ steps, we must hit the \gray\ path and the first vertex we come across is $z$. Let $c$ be the node before $z$ in the path. By the way we process $z$'s out-neighbor list, we first bring the $min$ to the front of the list, and then explore the other neighbor nodes in sequence, and hence $c$, the explored neighbor of $z$ must appear after $min$ in $z$'s list. 

For the converse, suppose $y$ is a white vertex. Either we never reach a \gray\ vertex in $max$ steps (and we will correctly determin its color in this case) or we reach $x$ or $x$'s ancestor $z$ from $y$ following through the (spurious) first vertices of the in-neighbor list of a \white\ vertex $y$. Note that the parent of a \white\ vertex is \white\ or \gray\ and it can never be \black. Hence $z$'s child in the path is \white. Hence that child will appear before $min$ in $z$'s list.
\end{proof}
Given the above test, if $y$ turns out to be white, the edge $(x,y)$ is added to the DFS tree, and $y$ now becomes the current \gray\ vertex. Note that maintaining the invariants are straightforward. Also, when any vertex $v$ has been completely explored by DFS, we retrieve its parent from the in-neighbor list to complete the backtracking step. This procedure is continued until DFS comes back to the source vertex $s$. We stop at this point. This is because, note that, our proof breaks down in the case when DFS in a directed graph produces a forest and some cross edges go across these trees. In that case, if we follow the path starting from $y$'s parent, we would reach the root of the DFS tree containing $y$ and this is different from the tree where $x$ belongs to. As we cannot maintain informations regarding all such roots of these previously explored DFS trees, we might spuriously conclude that $y$ is unvisited even though it is not the case. Thus our algorithm produces the DFS tree containing only the vertices reachable from the source vertex $s$ via some directed path in $G$. We leave open the case for desigining such logspace algorithm for the general directed graphs.

%As we do not know how to main maintain information regarding all roots that have been explored, we might spuriously conclude that $y$ is white. 
%Thus our algorithm produces the DFS tree containing only the vertices reachable from the source vertex $s$ 
%We leave open the case for designing such logspace algorithm for DFS in general directed graphs.

Given the above lemma, if $y$ turns out to be white, the edge $(x,y)$ is added to the DFS tree, and $y$ now becomes the current \gray\ vertex. Note that maintaining the invariants are easy. 
When any vertex $v$ has been completely explored by DFS, we retrieve its parent from the in-neighbor list to complete the backtracking step. This procedure is continued until DFS comes back to the source vertex $s$. We stop at this point and we have outputted all vertices reachable from $s$.

% This is because, note that, 
% Our proof breaks down in the case 
% even though it is not the case. 

To analyse the running time of our algorithm observe that testing for \gray\ takes at most $d$ steps. 
Testing for \black\ takes, in the worst case, at most $max$ steps to find the path, and at each step of the path, we take $d$ time to test whether the new node is \gray. Once we reach a \gray\ vertex, we spend at most $deg(z)$ steps to determine whether 
$c$ appears before $min$. Thus for each vertex in $x$'s list, we spend time proportional to $d + (d.max) + deg (z)$ time.
As $z$ (which is independent of $x$) can have degree at most $n$.
Thus, the overall runtime of the algorithm is
$\sum_{v\in V} deg(v) (d + d \ell + n)$ which is $O(m (n + (1 + \ell) \ell)$ which is $O(m (n+\ell^2))$, where $\ell$ is the maximum depth of DFS tree.

\section{BFS algorithms in the \rotated\ model}\label{bfs_rotate_appendix}

\subsection{BFS using $n+O(\lg n)$ bits--Proof of Theorem~\ref{bfsrotate}(a)} 
\label{subsec:rotatebfs}
It is well-known that BFS actually computes the shortest path lengths in unweighted undirected or directed graph $G$ from a given source vertex $s\in V$ to every vertex $v\in V$ that is reachable from $s$. I.e., if a vertex $v$ belongs to the $d$-th level in the BFS tree (assuming the root $s$ is at zero-th level), then we know that the length of the shortest path from $s$ to $v$ is $d$.
% See Theorem 22.5 \cite{CLRS} for a proof.
We use this crucially to design our BFS algorithms. We use a bit array $visited[1,\cdots,n]$ that stores a $0$ for an unvisited vertex, and $1$ for a visited vertex. We also maintain a counter $dist$ which stores the level of the vertex that is currently being explored in the BFS algorithm.
% By analogy with the recent BFS implementation by Banerjee et al. \cite{cocoon16}, a vertex that used to get a \black\ color is marked $1$ here but we are overloading the colors \grayone, \graytwo\ and \white\ to $0$. In what follows, we describe how we distinguish each color. First we deal with undirected graphs and later we mention the changes required to handle directed graphs. 

We start by setting $visited[s] = 1$, and initializing the counter $dist$ to $0$. At the next step, for every unvisited neighbor $v$ of $s$, we rotate their adjacency list so that $s$ appears as the first element in $v$'s list, set $visited[v] = 1$, and output $v$.
%iterate over each of the neighbors $v$ in the circular adjacency list of $s$ one by one, and use rotation to bring $s$ in the front of the circular adjacency list of $v$ without marking any other neighbor(s) of $v$ visited. This process stops when we exhaust all the neighbors of the starting vertex $s$. 
This step ensures that for each visited vertex, its parent is at the front of its adjacency list. We refer to this front element in the adjacency list of a visited vertex as its {\em parent pointer}. (Also, once we set the parent pointer for a vertex, we will not rotate its adjacency list in the remaining part of the algorithm.)
%We call this step as {\it setting up parent} phase as for all the vertices $v$ which are at distance $1$ from the root $s$, we have {\it fixed} their parent $s$ at the front location of their individual lists and this will remain so during the course of the entire algorithm. 
Once the root $s$'s list is fully processed as above, the $dist$ is incremented to $1$. The next step in the algorithm is to find all the vertices in the first level and mark all their unvisited neighbors as visited. As we haven't stored these vertices (in the first level), the challenge is to find them first. We use the following claim, to find the level number of a visited vertex.
The proof easily follows from the fact that the parent pointers are set for all the visited vertices, and that all the ancestors of a visited vertex are also visited.

\begin{claim}
\label{dclaim}
If the BFS has been explored till distance $d$ from the source, then for any $k \le d$, a vertex $x$ marked visited is in level $k$ if and only if we can reach the source $s$ in exactly $k$ steps by following through their parent pointers. Thus determining if a visited vertex is in level $d$ takes at most $d$ steps.
\end{claim}

%\begin{proof}
%Visited vertices have their parent set in their first locations, so vertices at distance $d$ can reach $s$ by following through their parent pointers for $d$ steps. Visited vertices in other levels can reach $s$ through a path of length $< d$.
%\end{proof}
So now we continue the BFS by scanning through the vertices, finding those vertices in level $d$ (using the above claim by spending $d$ steps for each vertex), and marking their unvisited neighbors visited, and making in their adjacency lists, their parent vertex as the first element, and incrementing $dist$. We stop our algorithm when we discover no new unvisited vertex while exploring any level. The correctness of the procedure and the space used by the algorithm are clear. To analyze the runtime, note that the time spent at level $d$ is $nd + \sum_{i \in V(d)} deg(i)$ where $V(d)$ is the set of vertices in level 
$d$ and $deg(i)$ is the degree of vertex $i$. Summing over all levels, we get a runtime $O(m + n \ell^2)$, where $\ell$ is the depth of the BFS tree.
%
%We note the following,
%{\bf Invariant}\label{inv2}: After the exploration of BFS,
%\begin{itemize}
%\item for every vertex $v$ (except root $s$ of the BFS tree), $v$ is marked $1$ in the visited array, and the first element in its adjacency list is its parent in the BFS tree.
%\end{itemize}
%Using the above invariant, we can also report any vertex to root (shortest) path in time proportional to the length of the path. 
%({\bf check this paragraph for correctness}) 
To handle directed graphs, we follow the outneighbor list as we go down, and we set up the parent at the first position in the in-neighbor list of every vertex $v$. To verify if $v$ belongs to the $d$-th level, we take $d$ steps from $v$ by following the parent pointers in the in-neighbor lists of the (visited) vertices along the path, and check if we reach $s$ at the end. This completes the proof of Theorem~\ref{bfsrotate}(a).
%\begin{theorem} \label{bfs1}
%Given a directed or undirected graph $G$ with $n$ vertices and $m$ edges, a BFS traversal of $G$ can be performed in $O(m+n\ell^2)$ time in the \rotated\ model
%using $n+O(\lg n)$ bits, where $\ell$ is the depth of the BFS tree of $G$. 
%\end{theorem} 

\subsection{BFS using $O(\lg n)$ bits--Proof of Theorem~\ref{bfsrotate}(b)}\label{subsec:rotatebfs-logn}
To reduce the space to $O(\lg n)$ bits, we dispense with the color array and explain how to determine visited and unvisited vertices. Assume that we can determine this in constant time.
%for BFS albeit using more time. The main challenge is to determine the visited vertices, which is needed both to identify vertices in the current level $d$, and to figure out which of their neighbors are visited. Since we do not store an explicit $visited$ array, visiting a vertex in this case corresponds to setting its parent pointer.
Our first observation is that Claim~\ref{dclaim} is true even for unvisited vertices even though 
the first vertex in the adjacency list of unvisited vertices can be an arbitrary vertex (not necessarily referring to their parent in the BFS tree). 
However, we know (by the property of BFS) that no unvisited vertex is adjacent to a vertex in level less than $d$, and hence they can not reach $s$ by a path at most $d$. Using the same argument, we can show that
\begin{claim}
\label{claim2}
If vertices up to level $d$ have been explored and visited, then a vertex $x$ is a visited vertex if and only if by following through the parent pointers, $x$ can reach $s$ by a path of length at most $d$. Furthermore, a vertex is in level $d$ if and only if we can reach $s$ by a path of length exactly $d$ by following through the parent pointers.
\end{claim}

Thus, to check whether a vertex is visited, we spend $O(d)$ time when exploring vertices at level $d$ instead of $O(1)$ time when the $visited$ array was stored explicitly.
%Now, if we implement the algorithm of last section (without the $visited$ array), for every neighbor of every vertex in each level (to determine whether it is a visited vertex to set parent) using the above claim, 
Hence, the total the time spent at level $d$ is $O(n d + d \sum_{i \in V(d)} deg(i))$, where the first term gives the time to find all the visited vertices at level $d$, and the second term gives the time to explore those vertices (i.e., going through their neighbor lists, identifying the unvisited neighbors and setting their parent pointers). Summing over all levels, we get a total runtime of $O(n \ell^2 + m \ell)$. Note that this algorithm works only when the input undirected graph is connected as Claim~\ref{claim2} breaks down if there are more than one component. The modifications to handle the directed graphs are similar to those for the directed graph BFS algorithm.
% Here we need the assumption that the source vertex $s$ reaches every other vertex of the directed graph $G$, and the rest of the algorithm follows similarly this one. 
This proves Theorem~\ref{bfsrotate}(b).

\section{Simulation of algorithms for \rotated\ model in the \implicit\ model}
\label{subsec:Rotate-Implicit}
The following result captures the overhead incurred while simulating 
any \rotated\ model algorithm in the \implicit\ model. Most of our algorithms in the \implicit\ model use these simulations often with some enhancements and tricks to obtain better running time bounds for some specific problems. 

%We provide the  proof of Theorem~\ref{simulation_result} in Appendix~\ref{simulation_appendix}.
\begin{theorem}\label{simulation_result}
Let $D$ be the maximum degree of a graph $G$. Then any algorithm running in $t(m,n)$ time in the \rotated\ model can be simulated in the \implicit\ model in (i) $O(D \cdot t(m,n))$ time when $G$ is given in an adjacency list, and (ii) $O(\lg D \cdot t(m,n))$ time when $G$ is given in an adjacency array. Furthermore, let $r_v(m,n)$ denote the number of rotations made in $v$'s (whose degree is $d_v$) list, and $f(m,n)$ be the remaining number of operations. Then any algorithm running in $t(m,n)=\sum_{v \in V} r_v(m,n) + f(m,n)$ time in the \rotated\ model can be simulated in the the \implicit\ model in (i) $O(\sum_{v \in V} r_v(m,n) \cdot d_v + f(m,n))$ time when $G$ is given in an adjacency list, and (ii) $O(\sum_{v \in V} r_v(m,n) \lg d_v + f(m,n))$ time when $G$ is given in an adjacency array. 
\end{theorem}

%\section{Proof of Theorem~\ref{simulation_result}--simulation of algorithms for \rotated\ model in the \implicit\ model}\label{simulation_appendix}
\begin{proof}
We can implement a single rotate operation that moves the head of the list by one position (which is assumed to be a unit-cost operation in the \rotated\ model) in the 
\implicit\ model by moving all the elements in the adjacency list circularly. If $d_v$ is the degree of a vertex $v$, then 
performing a rotation of the adjacency list of $v$ can be implemented in $O(d_v)$ time in the \implicit\ model. Thus, if we have 
an algorithm in the \rotated\ model that takes $t(m,n)$ time, then it can be implemented in $O(D \cdot t(m,n))$ time in the
\implicit\ model, where $D$ is the maximum degree of the graph. One can get a better runtime by analysing the algorithm 
for the \rotated\ model more carefully. In particular, if the runtime 
of the algorithm in the \rotated\ model can be expressed as $r(m,n) + f(m,n)$, where $r(m,n)$ is the number of rotations
performed and $f(m,n)$ is the remaining number of operations, then the algorithm can be implemented in the \implicit\ model in $O(D \cdot r(m,n)) + f(m,n)$ time.
Furthermore, if $r(m,n) \le \sum_{v \in V} r_v(m,n)$ where $r_v(m,n)$ is the number of rotations made in $v$'s list, then the runtime of the algorithm in the \implicit\ model can be 
bounded by $O(\sum_{v \in V} r_v(m,n) \cdot d_v + f(m,n))$.

If the input graph is given in the adjacency array representation  and if the ordering of the elements in the adjacency 
array is allowed to be changed, then one can simulate the rotate operation even faster.  
The main idea is to simulate the algorithm for the \rotated\ model after sorting the adjacency arrays of each vertex. Using
an in-place linear time radix sorting algorithm~\cite{FranceschiniMP07}, sorting all the adjacency arrays can be done in $O(m)$ time. Next, we
note that in the \rotated\ model, the head points to an element in an arbitrary position (called the front element) in the adjacency list of any vertex and the unit operation moves it one position. Thus, it is enough 
to show how to access the next element in sorted order. We maintain the following invariant: if the first element is the
$i$-th element in sorted order, then the adjacency array consists of the sorted array of all the adjacent vertices, with 
the first element swapped with the $i$-th element. To bring the $(i+1)$-st element to the front, we first perform a binary
search for the $i$-th element which is in the first position in the `almost sorted' adjacency array to find the position $i$, and move the elements appropriately to 
maintain the invariant (with the $(i+1)$-st element at the front). This takes $O(\lg d_v)$ time to simulate the rotation
of the adjacency array of a vertex $v$ with degree $d_v$. Thus, if we have an algorithm in the \rotated\ model that takes 
$t(m,n)$ time, then it can be implemented in $O(\lg D \cdot t(m,n))$ time in the \implicit\ model. Moreover, if the runtime 
of the algorithm in the \rotated\ model can be expressed as $\sum_{v \in V} r_v(m,n) + f(m,n)$, where $r_v(m,n)$ is an 
upper bound on the number of rotations performed on vertex $v$ and $f(m,n)$ is the remaining number of operations, then the algorithm can be implemented in the \implicit\ model in 
$O(\sum_{v \in V} r_v(m,n) \lg d_v + f(m,n))$ time.
\end{proof}

\section{DFS algorithms in the \implicit\ model---proof of Theorem~\ref{dfsimplicit}}
\label{subsec:implicitdfs}
%\paragraph{DFS using $O(\lg n)$ bits. }
To obtain a lex-DFS algorithm, we implement the $O(\lg n)$-bit DFS 
algorithm in the \rotated\ model, described in Section~\ref{subsec:dfsrotate-logn}, with a 
simple modification. 
First, note that in this algorithm (in the \rotated\ model), we bring the parent of a vertex to the 
front of its adjacency list (by performing rotations) when we visit a vertex for the first time. 
Subsequently, we explore the remaining neighbors of the vertex in the left-to-right order.
Thus, for each vertex, if its parent in the DFS were at the beginning of its adjacency list, then 
this algorithm would result in a lex-DFS algorithm. 
Now, to implement this algorithm in the \implicit\ model, whenever we need to bring the parent 
to the front, we simply bring it to the front without changing the
order of the other neighbors. Subsequently, we simulate each rotation by moving all the elements
in the adjacency list circularly. As mentioned in Section~\ref{subsec:Rotate-Implicit}, this results in 
an algorithm whose running time is $O(\sum_{v \in V} d_v (d_v + \ell) \cdot d_v) = O(m^3/n^2 + \ell m^2/n)$ if the graph is given in an adjancecy list and in $O(\sum_{v \in V} d_v (d_v + \ell) \cdot \lg d_v)  = O(m^2 (\lg n) / n + m \ell \lg n))$ when the graph is given 
in the form an adjacency array. This proves Theorem~\ref{dfsimplicit}(a) for undirected graphs. The results for the directed case follow from simulating the corresponding results for the directed graphs.%We can improve the running time and space bounds significantly by using the encoding bit tricks ideas and the resulting algorithms are discussed in details 

To prove the result mentioned in Theorem~\ref{dfsimplicit}(b), we implement the linear-time DFS algorithm of 
Theorem~\ref{dfsrotate} for the \rotated\ model that uses $n+O(\lg n)$ bits. 
%As described in Section~\ref{subsec:Rotate-Implicit},
This results in an algorithm that runs in $O(\sum_{v \in V} d_v^2 + n)  = O(m^2 / n)$ time 
(or in $O(\sum_{v \in V} d_v \lg d_v + n)  = O(m \lg m + n)$ time, when the graph is given 
as an adjacency array representation), using $n+O(\lg n)$ bits.
We reduce the space usage of the algorithm to $O(\lg n)$ bits by encoding the visited/unvisited 
bit for each vertex with degree at least $2$ within its adjacency list (and not maintaining this bit 
for degree-1 vertices). 
%But the running time will increase (compared to the algorithm in the \rotated\ 
%model) since we can not implement the rotations (which are assumed to be constant-time operations 
%in the rotated list model) efficiently in the \implicit\ model. 
We describe the details below.

Whenever a node is visited for the first time in the algorithm for the rotated list model, we bring
its parent to the front of its adjacency list. In the remaining part of the algorithm, we process
each of its other adjacent vertices while rotating the adjacency list, untill the parent comes to the
front again. Thus, for each vertex $v$ with degree $d_v$, we need to rotate $v$'s adjacency
list $O(d_v)$ times. 
%The rotation of the adjacency list of a vertex $v$, which is assumed to be a constant time 
%operation in the rotated lisrotated list model, can be implemented in $O(d_v)$ time in the \implicit\ model.
In the \implicit\ model, we also bring the parent to the front when a vertex is visited for the
first time, for any vertex with degree at least $3$. We use the second and third elements in the
adjacency list to encode the visited/unvisited bit. But instead of rotating the adjacency list 
circularly, we simply scan through the adjacency list from left to right everytime we need to 
find the next unvisited vertex in its adjacency list. This requires $O(d_v)$ time for a vertex 
$v$ with degree $d_v$. We show how to handle vertices with degree at most $2$ separately. 

As before, we can deal with the degree-1 vertices without encoding visited/unvisited bit as we 
encounter those vertices only once during the algorithm.  
%We also do not encode the visited/unvisited bit for degree-2 vertices. Instead, whenever 
%a degree-2 vertex $v$ is visited for the first time, we bring its parent $p$ (in the DFS tree) 
%to the front of its adjacency list. The other vertex in its adjacency list, $u$, is either an 
%ancestor (if the edge $v$-$u$ is a back edge) or its child in the DFS tree. If $u$ is an 
%ancestor, then $v$ is a leaf, and we immediately backtrack from $v$ to $p$ in the DFS 
%algorithm. If $u$ is a child of $v$ , then we follow the sequence of degree-2 vertices 
%untill we reach a vertex $y$ with degree at least $3$. 
For degree-2 vertices, we initially (at preprocessing stage) encode the bit $0$ using the two
elements in their adjacency arrays - to indicate that they are unvisited. When a degree-2 
vertex is visited for the first time from a neighbor $x$, we move to its other neighbor
-- continuing the process as long as we encounter degree-2 vertices until we reach a 
vertex $y$ with degree at least $3$. If $y$ is already visited, then we output the path
consisting of all the degree-2 vertices and backtrack to $x$. If $y$ is not visited yet,
then we output the path upto $y$, and continue the search from $y$, and after marking 
$y$ as visited. In both the cases, we also mark all the degree-2 nodes as visited (by 
swapping the two elements in each of their adjacency arrays).

During the preprocessing, for each vertex with degree at least $3$, we ensure that the 
second and third elements in its adjacency list encode the bit $0$ (to mark it unvisited).
We maintain the invariant that for any vertex with degree at least $3$, as long as it 
is not visited, the second and third elements in its adjacency array encode the bit $0$; and
after the vertex is visited, its parent (in the DFS tree) is at the front of its adjacency array, 
and the second and third elements in its adjacency array encode the bit $1$.
Thus, when we visit a node $v$ with degree at least $3$ for the first time, we bring its 
parent to the front, and then swap the second and third elements in the adjacency list, 
if needed, to mark it as visited. 
The total running time of this algorithm is bounded by $\sum_{v \in V} d_v^2 = O(m^2/n)$.
%where $d_v$ is the degree of vertex $v$.
%
%\begin{theorem}
%Given a directed or an undirected graph $G$ with $n$ vertices and $m$ edges, a DFS 
%traversal of $G$ can be performed in $O(m^2/n)$ time using $O(\lg n)$ bits.
%\end{theorem}

% If the adjacency lists are given in sorted order, then we can improve the running time to $O(m \lg m)$ as follows.
We can implement the above DFS algorithm even faster when the input graph is given in an 
adjacency array representation. We deal with vertices with degree at most $2$ exactly as 
before. For a vertex $v$ with degree at least $3$, we bring its parent 
to the front and swap the second and third elements to mark the node as visited (as before) 
whenever $v$ is visited for the first time. We then sort the remaining elements, if any, in the 
adjacency array, in-place (using the linear-time in-place radix sort algorithm~\cite{FranceschiniMP07}), and implement
the rotations on the remaining part of the array as described in Section~\ref{subsec:Rotate-Implicit}.
The total running time of this algorithm is bounded by $\sum_{v \in V} d_v \lg d_v = O(m\lg m +n)$.
This completes the proof of Theorem~\ref{dfsimplicit}(b).

\section{BFS algorithms in the \implicit\ model---proof of Theorem~\ref{bfsimplicit}}\label{omit_proof}
Before getting into the technical details, we first outline the main ideas involved in proving Theorem~\ref{bfsimplicit}. 
%See Appendix~\ref{app:bfsimplicit} for more details.
One can simulate the BFS algorithm of Theorem~\ref{bfsrotate}(b) (for the \rotated\ model) in the \implicit\ model using the simulation described in Section~\ref{subsec:Rotate-Implicit}. 
Since these BFS algorithms scan through each of the adjacency lists/arrays at most twice during the algorithm,
there won't be any slowdown in the runtime. This results in an algorithm running in $O(m \ell + n \ell^2)$ time, using $O(\lg n)$ bits. To improve the running time, we simulate the algorithm of Theorem~\ref{bfsrotate}(a) using the trick of encoding the $visited$ bit of each vertex in its adjacency list instead of storing the $visited$ array explicitly. This requires special attention to degree-1 and degree-2 vertices along with few other technical issues which are dealt in the proof given next.

\begin{proof}
Here we give the full proof of Theorem~\ref{bfsimplicit}. In particular, we provide all the details of the case when the degree of each vertex is at least $3$, and in that case, we show that we can implement the BFS algorithm using $4$ colors of~\cite{BCR}, by encoding the $4$ color of a vertex using the first three elements in its adjacency list, resulting in an algorithm that takes $O(m + n\ell)$ time. Moreover, when the degree of every vertex is at least $2 \lg n + 3$, then we show that the above algorithm can be implemented more efficiently, resulting in an algorithm that takes $O(m)$ time. Details follow. One can simulate the BFS algorithm of in Item~2 of Theorem~\ref{bfsrotate} (for the \rotated\ model) in the \implicit\ model using the simulation described in Section~\ref{subsec:Rotate-Implicit}. 
Since these BFS algorithms scan through each of the adjacency lists/arrays at most twice during the algorithm,
there won't be any slowdown in the runtime. This results in a BFS algorithm that runs in $O(m \ell + n \ell^2)$ time, using $O(\lg n)$ bits.

To improve the running time further, we simulate the algorithm of Theorem~\ref{bfsrotate}(a).
But instead of storing the $visited$ array explicitly, we encode the $visited$ bit of each vertex in its adjacency list, 
as explained below, resulting in an algorithm that takes $O(m + n \ell^2)$ time, using $O(\lg n)$ bits.
%
%{\bf Implementing the $n+o(n)$ bits \rotated\ algorithm in \implicit\ using bit encoding?}
%We can do this, but it also requires the assumption that the degree is at least 3 (since we need 
%to bring the parent to the front, and then use the next two vertices to store the visited bit).
%(Remove this paragraph.) May be this is easy to fix. For degree-2 nodes, we only store the visited bit.
%When we need to follow the parent pointers starting from a degree-2 node, we simply try both the neighbors.
%(Intermediate degree-2 neighbors simply take the "other" vertex - to follow the parent pointer.)
%See the description below.
%
For a vertex $x$ with degree at least $3$, we encode its $visited$ bit using the second and third elements in its adjacency list.
To set the parent pointer for $x$ (when it is visited for the first time), we bring its parent to the front, and move the second and third elements,
if necessary, to encode the $visited$ bit. We now describe how to deal with the degree-1 and degree-2 vertices.

First, observe that in the original BFS algorithm, we can simply output any degree-1 vertex, when it is visited for the first time. 
Thus, we need not store the $visited$ bit  for degree-1 vertices. For degree-2 vertices, we encode the $visited$ bit using the two neighbors.
We do not bring its parent to the front, as we do for other vertices. Whenever we need to check whether a visited degree-2 vertex is at
depth $d$, we follow parent pointers from both the neighbors - if one of them does not exist, then the other one is its parent - for a distance
of length at most $d$. (While following the parent pointers from a vertex to the root, it is easy to find its parent - since there is only one alternative to follow.) 
Thus, the first part of the theorem follows from Theorem~\ref{bfsrotate}(a), with the above modification.

To improve the runtime further, we implement the BFS algorithm using $4$ colors by~\cite{BCR}, where all the unvisited vertices are \white, all the visited, yet unfinished vertices of the two consecutive layers of BFS are \grayone\ and \graytwo\ respectively, and all the vertices which are completely explored are \black. 
Suppose the degree of every vertex in the given graph is at least three. In this case, we can encode the color of any vertex by permuting the first three elements in its adjacency array appropriately. This enables us to retrieve (or modify) the color of any vertex in constant time by reading (or modifying) the permuted order of the first elements in its adjacency array.

Since we haven't stored the \grayone\ or \graytwo\ vertices in a queue (as in the standard BFS), we scan through the entire vertex set (in the increasing order of their labels), and when we find any vertex $v$ colored \grayone, we color all its \white\ neighbors with \graytwo, and color $v$ itself with \black. We call this an {\em exploration phase}. The time taken for each exploration phase is equal to the sum of the degrees of all the \grayone\ vertices at the beginning of the phase, plus $O(n)$. At the end of each exploration phase, we change the colors of all \graytwo\ vertices to \grayone\ and also output them, using an additional $O(n)$ time. We call this a {\em consolidation phase}. We need to repeat the exploration and consolidation phases for $\ell$ times before all the nodes are colored \black, where $\ell$ is the height of the BFS tree. Thus the overall time taken by this procedure can be bounded by $O(m + n\ell)$. 
This proves the second part of the theorem.
%Thus we have
%\begin{theorem}
%Given a directed or an undirected graph $G$ with $n$ vertices and $m$ edges, a BFS traversal 
%of $G$ can be performed in $O(m+n\ell)$ time using $O(\lg n)$ bits where $\ell$ is the depth 
%of the BFS tree of $G$ in the \implicit\ model, if the degree of every vertex is at least $3$.
%\end{theorem}

%If we have $O(s \lg n)$ bits of additional space, we can reduce the running time of the above 
%algorithm to $O(m + n^2/s)$, by storing storing all the \grayone\ vertices in each phase whenever
%there are at most $s$ of them. Since the number of levels for which we can not store the \grayone\
%vertices for those levels is at most $O(n/s)$, we need to scan the vertex list at most $O(n/s)$
%times.

%\begin{theorem}
%Given a directed or an undirected graph $G$ with $n$ vertices and $m$ edges, a BFS traversal 
%of $G$ can be performed in $O(m + n^2/s)$ time using $O(s)$ bits.
%\end{theorem}
%
If every vertex has degree at least $2 \lg n + 3$, then we can perform BFS in $O(m)$
time (note that $m \ge n \lg n$ in this case) -- by encoding the current set of \grayone\ 
vertices as a linked list by using $2 \lg n$ vertices to encode (a pointer to) the next 
vertex in the list. The time to read all the \grayone\ vertices in a phase when there are $k$ 
vertices colored \grayone\ becomes $O(k \lg n)$ instead of $O(n)$. This results in $O(m + n \lg n)$ time 
which is $O(m)$.  This proves the third part of the theorem.
\end{proof}
%
%Thus we have
%\begin{theorem}
%Given a directed or an undirected graph $G$ with $n$ vertices and $m$ edges, a BFS traversal 
%of $G$ can be performed in $O(m)$ time using $O(\lg n)$ bits, if the degree of every
%vertex is at least $2 \lg n + 3$.
%\end{theorem}

\section{Minimum Spanning Tree}\label{mstproof}
In this section, we start by giving an in-place implementation of the Prim's algorithm~\cite{CLRS} to find a minimum spanning tree of a given weighted undirected graph in the \rotated\ model. Here we are given a weight function $w: E \rightarrow Z$. We also assume that the weights of non-edges are $\infty$ and that the weights can be represented using $O(\lg n)$ bits. The input representation also changes slightly to accommodate these weights. Now each element of the circular linked list has three fields, (a) the vertex label, (b) the weight, and (c) the pointer to the next vertex respectively. In what follows, when we talk about computing a minimum spanning tree, what we mean is reporting the edges of such a tree. Our result is the following,

\begin{theorem}\label{mstproof1}
A minimum spanning tree of a given undirected weighted graph $G$
%, where the weights of any edge can be represented in $O(\lg n)$ bits, 
can be found using $O(\lg n)$ bits and in $O(mn)$ time in the \rotated\ model.
\end{theorem}

\begin{proof}
Our \rotated\ model algorithm basically mimics Prim's algorithm with a few tweaks. Prim's algorithm starts with initializing a set $S$ with a vertex $s$. For every vertex $v$ not in $S$, it finds and maintains $d[v] = \min \{ w(v,x): x \in S\}$ and $\pi[v]=x$ where $w(v,x)$ is the minimum among $\{ w(v,y): y \in S\}$. Then it repeatedly deletes the vertex with the smallest $d$ value from $V-S$ adding it to $S$. Then the $d$ values are updated by looking at the neighbors of the newly added vertex. To implement this procedure using just extra $O(\lg n)$ bits of space, first we need to find a way to mark/unmark a vertex $v$ if it has been taken into $S$ without using $n$ bits explicitly. The way we do this is as follows. In the preprocessing step, the algorithm spends $O(m)$ time to bring the minimum valued neighbor (denote it by \textit{min}) in the neighbor list of every vertex $v$ by rotation. Subsequently we would attach the following meaning with the position of \textit{min} in the list of any vertex $v$. If the first element in the list of any vertex $v$ is \textit{min}, this means that $v$ is not taken into $S$ so far during the execution of the algorithm, otherwise it belongs to $S$. This way we can store the information regarding the status of any vertex $v$ without using any extra space but actually figuring out this information takes time proportional to the degree of $v$. Note that for vertices having degree one, we cannot determine exactly its status correctly by this method. But a simple fact we can use here. If a vertex $z$ has degree one (say its neighbor is $y$), then the edge $(y,z)$ is always a part of the minimum spanning tree. Hence, after the preprocessing step, we can output all such edges at once and embark on executing the rest of the algorithm.

The algorithm initializes the set $S$ with the starting vertex $s$, and goes to its list to give a rotation so that \textit{min} does not stay in the first location in $s$'s list. We call this step as the \textit{marking} step. This is followed by finding the smallest weight neighbor (say $u$) of $s$. According to Prim's algorithm, $u$ should now move to $S$. We achieve the same by marking $u$ i.e., going to $u$'s list to give a rotation to move \textit{min} from the first location to indicate that $u$ belongs to $S$ now and subsequently continue to in $u$'s list find its smallest weight neighbor and repeat. Thus at any generic step of the algorithm, we go over the list of unmarked vertices (i.e., those vertices having \textit{min} at the first position of their respective lists)  and collect the minimum weight vertex (say $t$), and $t$ is then marked and we continue until all the vertices are marked.

Clearly this method returns all the correct edges of a minimum spanning tree of $G$. Space bound of this algorithm is easily seen to be $O(\lg n)$ bits for keeping a few variables. Total time spent by algorithm can be bounded by $O(mn)$ where at the preprocessing step, it spends $O(m)$ time and after the preprocessing, for reporting each edges of the minimum spanning tree, in the worst case, the algorithm spends $O(m)$ time, hence $O(mn)$ is the bound on the running time. 
\end{proof}

As mentioned in Section~\ref{subsec:Rotate-Implicit}, simulating this algorithm in the \implicit\ model would result in an algorithm having running time $O(mn. d_v)$=$O(mn^2)$ if the graph is given in an adjacency list and in $O(mn. \lg d_v)$=$O(mn \lg n)$ when the graph is represented in an adjacency array. Hence, we have the following,

\begin{theorem}\label{mstproof2}
In the \implicit\ model a minimum spanning tree of a given undirected weighted graph $G$
%, where the weights of any edge can be represented in $O(\lg n)$ bits, 
can be found using $O(\lg n)$ bits and 
\begin{enumerate}[noitemsep,topsep=0pt]
\item $O(mn^2)$ time if the graph is given in an adjacency list, and
\item $O(mn \lg n)$ time when the graph is represented in an adjacency array.
\end{enumerate}
\end{theorem}

\section{Consequences}
In this section, we provide some applications/consequences of our results that we derived in the earlier sections.

\subsection{Improved algorithm for DFS and applications in ROM}
\label{romdfs}
In this section, we show, using a little more space, how to simulate any \rotated\ model algorithm in the read-only model. This results in improved space-efficient algorithms for various fundamental graph problems in the read-only model. 
Observe that, the only modification of the input that we do in our \rotated\ model 
algorithm is to make the head pointer  point to an arbitrary element in the adjacency list (instead of a fixed element) at various times. To simulate this in read-only memory, we can simply maintain a pointer in each of the lists in the adjacency list. The resources required to store and update such pointers is proven in the following lemma \cite{ChakrabortyRS16}. We provide the proof here for completeness.

\begin{lemma}[\cite{ChakrabortyRS16}]\label{lem:adjlist-pointers}
Given the adjacency list representation of a directed or an undirected graph $G$ on $n$ vertices with $m$ edges,
using $O(m)$ time, one can construct an auxiliary structure of size $O(n \lg (m/n))$ bits that can 
store a ``pointer'' into an arbitrary index within the adjacency list of each vertex. Also, updating 
any of these pointers (within the adjacency list) takes $O(1)$ time.
\end{lemma}

\begin{proof}
We first scan the adjacency list of each vertex and construct a bitvector $B$ 
as follows: starting with an empty bitvector $B$, for $1 \le i \le n$, if $d_i$ is 
the length of the adjacency array of vertex $v_i$ (i.e., its degree), then we append the string 
$0^{\lceil{\lg d_i}\rceil -1}1$ to $B$. The length of $B$ is $\sum_{i=1}^n \lceil{\lg d_i}\rceil$, 
which is bounded by $O(n \lg (m/n))$. We construct auxiliary structures to support $select$ queries 
on $B$ in constant time~\cite{Munro96}. %, using Theorem~\ref{staticbit}.
We now construct another bitvector $P$ of the same size as $B$, which stores the pointers into the adjacency array of each vertex. The pointer into the adjacency array of vertex $v_i$ is stored using the $\lceil{\lg d_i}\rceil$ bits in $P$ from 
position $select(i-1,B)+1$ to position $select(i,B)$, where $select(0,B)$ is defined to be $0$. Now, using 
select operations on $B$ and using constant time word-level read/write operations, one can access and/or 
modify these pointers in constant time.
\end{proof}

To actually get to the element in the list, we need the graph to be represented as what is referred as adjacency array~\cite{ElmasryHK15}. Here given an index in the list of a vertex, we can access the (adjacent) vertex in that position of the vertex's adjacency list in constant time. Now if we simulate our \rotated\ model algorithm of Theorem \ref{dfsrotate} in read-only memory using the auxiliary structure as stated in Lemma \ref{lem:adjlist-pointers} as additional storage, then we obtain,
\begin{theorem}\label{thm:read-only}
A DFS traversal of an undirected or a directed graph $G$, represented by an adjacency array, on $n$ vertices and $m$ edges can be performed in $O(m+n)$ time using $O(n \lg (m/n))$ bits, in the read-only model.
\end{theorem}

The above result improves the DFS tradeoff result of
Elmasry et al.~\cite{ElmasryHK15} for relatively sparse graphs in the read-only model. In particular, they showed the following,

\begin{theorem}[\cite{ElmasryHK15}]\label{elmasry}
For every function $t:\mathbb{N}\rightarrow \mathbb{N}$ such that $t(n)$ can be computed within the resource bound of this theorem (e.g., in $O(n)$ time using $O(n)$ bits), the vertices of a directed or undirected graph $G$, represented by adjacency arrays, with $n$ vertices and $m$ edges can be visited in depth first order in $O((m+n)t(n))$ time with $O(n+n\frac{\lg \lg n}{t(n)})$ bits.
\end{theorem}

Thus to achieve $O(m+n)$ time for DFS, their algorithm (Theorem \ref{elmasry}) uses $O(n \lg\lg n)$ bits. This is $\Omega(n \lg (m/n))$ for all values of $m$ where $m = O(n \lg n)$. 
% Hence, for sparse graphs we obtain better tradeoff bounds. 
Banerjee et al. \cite{BCR} and Kammer et al. \cite{KammerKL16} recently provided another DFS implementation taking $O(m+n)$ bits of space and runs in $O(m+n)$ time. Note that, Theorem \ref{thm:read-only} improves the space bound of the above mentioned DFS implementations from $O(m+n)$ space to $O(n \lg (m/n))$, while maintaining the same linear running time. Chakraborty et al. \cite{ChakrabortyRS16} obtained similar results by a slightly different technique. In what follows, we show that using Theorem \ref{thm:read-only}, we can 
improve the space bounds of some of the classical applications of DFS in read-only model.

To illustrate this, note that, one of the many classical applications of DFS (see \cite{CLRS}) include (i) topological sorting of the vertices of a directed acyclic graph \cite{Knuth73a}, (ii) producing a sparse (having $O(n)$ edges) spanning biconnected subgraph of a undirected biconnected graph $G$ \cite{Elmasry10}, and (iii) given an undirected $2$-edge-connected graph $G=(V,E)$, to orient each edge $(u,v)\in E$ as an arc $(u,v)$ or $(v,u)$ to obtain $D=(V,A)$ such that $D$ becomes strongly connected. For all of these problems, classical algorithms \cite{Knuth73a,Elmasry10} take linear time and $O(n\lg n)$ bits of space. Recently, Banerjee et al. \cite{BCR} showed the following,

\begin{theorem}\label{bcr}\cite{BCR}
In the read-only model, if the DFS of $G$ on $n$ vertices and $m$ edges, can be performed in $t(m,n)$ time using $s(m,n)$, where $s(m,n)=\varOmega(n)$, bits of space, then using $O(s(m,n))$ bits and in $O(t(n,m))$ time, we can output
\begin{enumerate}[noitemsep,topsep=0pt]
% \item when $G$ is a directed graph, then we can output the strongly connected components of $G$ in $O(\frac{t(m,n)n\lg n}{s(m,n)})$ time using $O(s(m,n))$ bits,
 \item the vertices of a directed acyclic graph in topologically sorted order,
 \item the edges of a sparse spanning biconnected subgraph of a undirected biconnected graph $G$, and
 \item a strongly connected orientation of a undirected $2$-edge-connected graph $G$.
\end{enumerate}
\end{theorem}

%Proof details of Theorem \ref{bcr} is provided in the appendix. 
Now plugging the improved DFS algorithm of Theorem \ref{thm:read-only} in the above theorem, we obtain for all the applications an improved (over the classical implementation) space-efficient implementations taking $O(m+n)$ time and $O(n \lg (m/n))$ bits of space. 

\subsection{Space-efficient approximation algorithms using Baker's approach}\label{baker}
In this section we present in-place Logspace Approximation Scheme or LSAS ($(1\pm \epsilon)$ approximation algorithm running in logspace for every $\epsilon> 0$) for a class of MSO-definable optimization problems which are amenable to the Baker's method \cite{Baker94}, also known as \emph{shifting technique}, in locally bounded treewidth graphs. 

There are two main general approaches for designing PTASs for problems on planar graphs. The first approach is based on planar separators~\cite{LiptonT80}. The approximation algorithms resulting from this approach are impractical in general. To address this, Baker~\cite{Baker94} introduced the second approach for PTASs in planar graphs, based on decomposition into overlapping subgraphs of bounded outerplanarity, which are of bounded treewidth. For a general account on these, see~\cite{Demaine2008}.

Baker's method was originally designed to give PTASs for a host of \NP -hard optimization problems on planar graphs like minimum vertex cover, minimum dominating set, maximum independent set etc which are hard to approximate in general graphs. Many of these remain \NP -hard even in planar graphs. Later the technique was generalized to a broader class of graphs called graphs of bounded local treewidth~\cite{Eppstein00, DemaineH04}. For a vertex $v$ of a graph $G$ and integer $k \ge 0$, by $G^k_v$ we denote the subgraph of $G$ induced by vertices within distance at most $k$ from $v$ in $G$. A class of graphs $ \mathcal{G}$ is of bounded local treewidth if there exists function $f$ such that for every graph $G \in \mathcal{G}$ and every vertex $v$ of $G$, treewidth$(G^k_v) \le f(k)$.

\subsubsection{Baker's Algorithm}
The main two computational bottlenecks in Baker’s approach are 
\begin{enumerate}[noitemsep,topsep=0pt]
\item decomposing the graph into bounded treewidth graphs, and 
\item solving the optimization problem on bounded tree width graphs optimally and combining these solutions. 
\end{enumerate}

Step (1) requires performing BFS on the graph $G$ and considering the induced subgraphs $G_{i,j}$ between layers $ki +j$ and $k(i+1) +j$ (which are of treewidth $O(k)$) for $i \ge 0$ and offset $0 \le j \le k-1$ by deleting (or including in both the adjacent slices) the vertices/edges in every $k= O(1/\epsilon)$-th BFS layer (particular details differ depending on the problem). By choosing the right offset, we can make sure this affects the optimum solution at most by $\epsilon$ factor. For Step (2), given an MSO-definable problem, using Boadlander's~\cite{Bodlaender96} and Courcelle’s theorem~\cite{CourcelleM93} on bounded treewidth graphs we can get the optimal solution. Baker's approach is highly efficient in terms of time complexity as this gives linear time approximation schemes but the scenario changes when considering the efficiency of the approach in terms of the space needed. 

Though we can use the result of~\cite{ElberfeldJT10} for the second part of the algorithm which proves logspace analogue of Boadlander's and Courcelle’s theorem, in the the first part we need to compute distance. BFS can be performed in \NL\ in general graphs and in $\UL \cap \coUL$ for planar graphs~\cite{TW10}. Since $\UL \subseteq \NL$ is not known (or even believed) to be inside \Log, the \NL\ bound ($\UL \cap \coUL$ for planar graphs) for Baker's algorithm is the best known in ROM~\cite{DK}. Recently~\cite{DKM} has given an LSAS for maximum matching in planar graphs and some more sparse graph classes (a problem amenable to Baker's method) but a logspace analogue of Baker's algorithm in general is not yet known. 

Since in our in-place models (both \rotated\ and \implicit\ ) we can overcome this hurdle by computing distance in \Log, we obtain the following result. Notice that though our in-place algorithms change the ordering of the vertices in an adjacency list (or array) the graph remains the same and so the distances between vertices.

\begin{theorem}
In both the \rotated\ and the \implicit\ model MSO-definable optimization problems which are amenable to Baker's method has an LSAS in locally bounded treewidth graphs. 
\end{theorem}

\subsection{Solving NP-hard problems in in-place models}\label{nphardlogspace}
We show how to solve some NP-hard graph problems using logspace and exponential time in the \rotated{} and \implicit{} models. In particular, this implies that problems such as vertex cover and dominating set can be solved in exponential time using $O(\lg n)$ bits  in both the models. In constrast, note that, no NP-hard problem can be solved in the ROM model using $O(\lg n)$ bits unless {\sf P=NP}.

Similar to Fomin et al.~\cite{FominGLS16}, we define a class of graph problems in {\sf NP} which we call {\it graph subset problems} where the goal is to find a subset of vertices satisfying some property. We show a meta theorem showing that a restricted class of graph subset problems that are in {\sf NP} admit log-space exponential algorithms in the \rotated{} and \implicit{} models. 

Given a graph $G$ with its adjacency list, we encode a subset of the vertices as follows. For every vertex in the subset, we bring in the minimum labelled vertex among its neighbors to the front of the list, and for others, we keep a vertex with a higher label (than the minimum) at the front of the list. So it takes a linear time to check whether a vertex is in the subset. The algorithm enumerates all subsets (by this encoding) and simply verifies using the {\sf NP} algorithm whether that subset satisfies the required property until it finds a subset satisfying the property or it has enumerated all the subsets. By a standard theorem in complexity theory~\cite{AroraB}, every problem in {\sf NP} is actually verifiable by a log-space ROM and hence the overall space taken by our algorithm is only logarithmic.
Note that our algorithm requires that the adjacency list of any vertex has at least two values, i.e. that the degree of any vertex is at 
least two. Thus we have
\begin{theorem}
Any graph subset problem in {\sf NP} can be solved using $O(\lg n)$ bits of extra space (and exponential time) in the \rotated{} and \implicit{} models in graphs $G$ having minimum degree $2$.
\end{theorem}

\noindent
{\bf Remark 1. } We remark that the above idea can work for other graph problems that are not necessarily subset problems. For example, for testing hamiltonicity, we can simply explore all neighbors of a vertex (starting at the smallest labelled neighbor so we know when we have explored them all) in a systematic fashion simply encoding them into the adjacency list by moving the current neighbor to the front of the list, and test whether together they form a cycle of length $n$.

If the graphs have larger minimum degree (say at least $2 \lg n$), we can even encode pointers in each adjacency list and using that we can even test for graph isomorphism in logarithmic space by explicitly maintaining the vertex mapping by these encoding pointers.

\noindent
{\bf Remark 2. } The minimum degree $2$ restriction in the above theorem is not a serious restriction as for many problems (like vertex cover, dominating set and traveling salesperson problem), there are preprocessing routines that can handle (not necessarily in our model) and eliminate degree $1$ vertices.

\section{Concluding remarks}\label{end}
Our initial motivation was to get around the limitations of ROM to  
obtain a reasonable model for graphs in which we can obtain space  
efficient algorithms. We achieved that by introducing two new frameworks and obtained efficient (of the order of $O(n^3 \lg n)$) algorithms using $O(\lg n)$ bits of space for fundamental graph search procedures. We also discussed various applications of our DFS/BFS results, and it is not surprising that many simple corollaries would follow as DFS/BFS being the backbone of so many graph algorithms. We showed that some of these results also translate to improved space efficient algorithms in ROM (by simulating the \rotated\ model algorithms in ROM with one pointer per list). With some effort, we can obtain log space algorithm for minimum spanning tree. These results can be contrasted with the state of the art results in ROM that take almost linear bits for some of these problems other than having large runtime bounds. All our algorithms are conceptually simple, and as they don't use any heavy data structures, we believe that they are also practical to implement. Still, there are plenty of algorithmic graph problems to be studied in these models. We believe that our work is the first step towards this and will inspire further investigation into designing in-place algorithms for other graph problems. One future direction would be to improve the running time of our algorithms to make them more practical.

Surprisingly we could design log-space algorithm for some {\sf P}-complete problems, and so it is important to understand the power of our models.  Towards that we discovered that we can even obtain log-space algorithms for some NP-hard graph problems. More specifically, we defined {\it graph subset problems} and obtained log-space exponential time algorithms for problems belonging to this class. One interesting future direction would be to determine the exact computational power of these models along with exploring the horizon of interesting complexity theoretic consequences of problems in these models.

%\bibliographystyle{plain}
%\bibliography{dfs}
%\newpage
%\section{Appendix} 
%\label{appendix}
%\appendix
%\input{consequence}
%\input{nphard}
%\input{Baker}
%\input{missing-proofs}
%\input{appendix.tex}

\end{document}